\documentclass[%
 reprint,
superscriptaddress,
nofootinbib,
 amsmath,amssymb,
 aps,
 prr,
longbibliography,
]{revtex4-2}
\usepackage{bbold}
\usepackage[breaklinks=true,colorlinks=true
,urlcolor=blue
,anchorcolor=blue
,citecolor=blue
,filecolor=blue
,linkcolor=blue
,menucolor=blue
,pagecolor=blue
,linktocpage=true
,pdfproducer=medialab
,pdfa=true
]{hyperref}
\usepackage{graphicx,hyperref,color}
\usepackage{color}
\usepackage[dvipsnames,svgnames]{xcolor}
\usepackage{dcolumn}
\usepackage{bm}
\usepackage{bbm}
\usepackage{mathtools}
\usepackage{amsmath,braket}
\usepackage{amssymb,physics}

\usepackage{tikz}
\usepackage{circuitikz}
\usetikzlibrary{shapes.geometric} 
\usetikzlibrary{3d,calc}
\usepackage{pgfplots}
\pgfplotsset{compat=1.17}

\usetikzlibrary{decorations.pathreplacing} 
\usetikzlibrary{decorations.markings}
\tikzset{snake it/.style={decorate, decoration=snake}}

\def\ep{{\epsilon}}

\def\frac#1#2{{#1\over #2}}

\def\s{\sqrt}

\def\Re{{\rm Re}}

\def\be{\begin{equation}}
\def\ee{\end{equation}}
\def\ba{\begin{eqnarray}}
\def\ea{\end{eqnarray}}

\def\ti{\tilde}
\def\ap{\alpha}

\def\ddd{\cdot\cdot\cdot}
\def\no{\nonumber \\}

\def\la{\langle}
\def\lb{\rangle}
\def\ep{\epsilon}

\def\vp{\varphi}

\allowdisplaybreaks[4]

\usepackage{amsthm}
\newtheoremstyle{break}
  {\topsep}  
  {\topsep}   
  {}          
  {}          
  {\bfseries} 
  {.}         
  {\newline}  
  {} 
\theoremstyle{break}
\newtheorem{dfn}{Definition}[section]

\newtheorem{thm}[dfn]{Theorem}

\begin{document}
\begin{flushright}
YITP-26-68
\\
\end{flushright}
\title{{\large Holographic Dual of PT-Symmetric BCFT}}

\author{Ryota Maeda}
\affiliation{Center for Gravitational Physics and Quantum Information, Yukawa Institute for Theoretical Physics, Kyoto University, Kitashirakawa Oiwakecho, Sakyo-ku, Kyoto 606-8502, Japan}

\author{Nanami Nakamura}
\affiliation{Center for Gravitational Physics and Quantum Information, Yukawa Institute for Theoretical Physics, Kyoto University, Kitashirakawa Oiwakecho, Sakyo-ku, Kyoto 606-8502, Japan}

\author{Tadashi Takayanagi}
\affiliation{Center for Gravitational Physics and Quantum Information, Yukawa Institute for Theoretical Physics, Kyoto University, Kitashirakawa Oiwakecho, Sakyo-ku, Kyoto 606-8502, Japan}
\affiliation{Inamori Research Institute for Science, 620 Suiginya-cho, Shimogyo-ku, Kyoto 600-8411, Japan}


\begin{abstract}
We present a holographic dual of a two dimensional conformal field theory with non-Hermitian but Parity-Time (PT) symmetric boundary conditions, by applying the AdS/BCFT duality and by introducing an imaginary valued scalar field localized on an end-of-the-world brane. We find that as we increase the strength of the non-Hermitian PT-symmetric interactions, the system experiences a spontaneous PT-symmetry breaking. We also consider its Wick rotated setup as a new quantum quenched state and show that its growth of entanglement entropy can be larger than the standard results obtained from usual Cardy states.
\end{abstract}

\maketitle




\section{Introduction} 

Even though the Hermitian property of a given Hamiltonian guarantees that the energy spectrum is real valued, it is known that the opposite is not true. Even if the Hamiltonian is not Hermitian $H^\dagger\neq H$, it can have a real valued spectrum if the Hamiltonian is pseudo Hermitian i.e. $H^\dagger=\tau H \tau^{-1}$ for a certain Hermitian operator $\tau$. A typical such example is the PT invariant quantum system, where $P$ is the parity and $T$ is the time reversal operator \cite{Bender:1998ke,Bender:2007nj}. 
An interesting class of such examples is the PT invariant spin chains, where we can observe remarkable phase transitions due to the spontaneous PT-symmetry breaking which change real valued energy spectra into complex ones \cite{Korff:2008xx,Fring:2012id,Kattel:2023ras,Kattel:2024lot}. Moreover, in the context of conformal field theory (CFT), we can find the PT invariant non-Hermitian systems in the non-unitary minimal models \cite{Castro-Alvaredo:2017udm,Fukusumi:2025fir,Fukusumi:2025xrj,Diatlyk:2026oxm}.
These and other interesting classes of examples, pioneered by \cite{Hatano_1996}, have opened up an intriguing research area of non-Hermitian quantum systems
\cite{Ashida:2020dkc}. Refer to \cite{Gardas:2016gja,Cao:2023mfy} for 
thermal or finite temperature aspects of such non-Hermitian systems. 
For articles related to PT-symmetry and the reality of energy spectrum, refer to \cite{Mostafazadeh:2002id, Bender:2005tb,  SCHOLTZ199274}.

Motivated by these successful developments in non-Hermitian quantum systems, it is intriguing to consider their gravity duals via the holography \cite{tHooft:1993dmi,Susskind:1994vu}, by expanding the examples of the AdS/CFT \cite{Maldacena:1997re,Gubser:1998bc,Witten:1998qj}.
If this is concretely realized, we can efficiently analyze strongly interacting non-Hermitian quantum systems from classical gravity calculations.
There have been several works in this direction to construct gravity duals of PT invariant CFTs \cite{Arean:2019pom,Arean:2024gks,Xian:2023zgu} (see also \cite{Faedo:2019nxw,Ghodrati:2025fah,Song:2022awq,Begines:2026mlv} for closely related works). In these works, typically bulk matter fields are chosen to be complex valued so that they are dual to the non-Hermitian deformations of the original Hermitian CFTs, which are translationally invariant. Due to the back-reactions via the bulk Einstein equation, the bulk AdS metric gets also complex valued in the PT broken phases. Since we still do not have a solid theory of holography with a complexified metric, this involves major challenges. Thus it is helpful to find well-controllable examples.

In this article, instead, we introduce the non-Hermitian interactions only at the boundaries $x=x_*$ and $x=-x_*$, described by a boundary primary operator $O(\tau,y_1,\ddd,y_{d-2})$,
keeping the $d$ dimensional spacetime to be Hermitian where a CFT is defined. 
The CFT action looks like
\ba
&& S_{CFT}=S^{(0)}_{CFT}+i\lambda\int_{x=x_*} d\tau d^{d-2}y O(\tau,y)\no
&&\ \ \ -i\lambda\int_{x=-x_*} d\tau d^{d-2}y O(\tau,y).
\label{BCFTaction}
\ea
Here $S^{(0)}_{CFT}$ is the action of a unitary CFT with its anti de Sitter space (AdS) dual. The deformation proportional to the real valued parameter $\lambda$ makes the Hamiltonian non-Hermitian. However this system is PT invariant because the Parity exchanges the boundary at $x=x_*$ with the other one at $x=-x_*$, while the time reversal, which is anti-linear, maps $i\lambda$ to $-i\lambda$. In quantum spin chains, such PT invariant models with the boundary non-Hermitian interactions have been considered as typical examples \cite{Korff:2008xx,Fring:2012id,Kattel:2023ras,Kattel:2024lot}, where spontaneous PT-symmetry breakings were observed. In two dimensional BCFTs, the simplest example will be a free $c=1$ massless scalar field $X$ on an interval $[0,\pi]$ with two Dirichlet boundary conditions. If we impose $X(0)=i\lambda$ and $X(\pi)=-i\lambda$, this gives a PT-symmetric non-Hermitian BCFT. However, in this example the energy spectrum is clearly real-valued and bounded from below $E\geq -\frac{\lambda^2}{2\pi^2}$, which is always PT-symmetric (see the Appendix \ref{ap:free}).

This motivates us to study this PT invariant system (\ref{BCFTaction}) because holography allows us to explore strongly coupled dynamics. In this paper, we will find a class of gravity duals under such deformations by employing the AdS/BCFT construction \cite{Takayanagi:2011zk,Fujita:2011fp,Karch:2000gx}. For this, we introduce the end-of-the-world brane (EOW brane) in AdS geometry so that the gravity dual of the BCFT is given by the region surrounded by the union of the asymptotic AdS boundary and the EOW brane. To describe the boundary scalar operator deformation (\ref{BCFTaction}) at $x=\pm x_*$, we introduce a scalar field localized on the EOW brane \cite{Kanda:2023zse,Kanda:2023jyi} and allow it to take imaginary values.

For simplicity, we will choose $d=2$ and assume that the primary operator $O(\tau)$ is exactly marginal. 
The latter assumption guarantees that the theory is given by the boundary conformal field theory (BCFT) as it preserves a part of conformal symmetry even at the boundaries. In the AdS/BCFT construction, we consider a three dimensional pure AdS gravity and insert the EOW brane with a massless scalar field $\phi$ localized on it, which couples to the boundary primary $O(\tau)$ at $x=\pm x_*$ as in (\ref{BCFTaction}). We impose the boundary condition of this scalar field as 
\ba
\phi(\tau,x_*)=i\lambda,\ \ \ \ \phi(\tau, -x_*)=-i\lambda,
\label{nhbc}
\ea
which makes $\phi$ imaginary valued.
This model has the advantage that even though we expect backreactions due to the EOW brane in general, we can always rewrite the three dimensional bulk metric into the pure AdS$_3$ one. Moreover, under a Wick rotation, this model provides a new class of quantum quenches as we will see later. 

\section{Holographic Model for PT invariant BCFT} 
We consider the two dimensional BCFT (\ref{BCFTaction}) at finite temperature $T=1/\beta$ described by the Euclidean flat coordinate $(\tau,x)$ on a cylinder, where $\tau$ is the Euclidean time coordinate.
We also write the width between the two boundaries by 
$\Delta x$. Its gravity model dual is constructed via the AdS/BCFT \cite{Takayanagi:2011zk,Fujita:2011fp,Kanda:2023zse} as follows. The gravity dual is given by a three dimensional bulk region $M$ which is surrounded by the AdS boundary $\Sigma$ and an end-of-the-world (EOW) brane $Q$. The BCFT lives on the two dimensional cylinder,  defined by the AdS boundary ($z=\ep$) of the three dimensional bulk $M$.
We also assume a massless scalar field $\phi$ localized on the EOW brane. 
The total gravity action reads as follows (EH and GHY are abbreviations of 
Einstein-Hilbert and Gibbons-Hawking-York):
\begin{equation}
  \begin{aligned}
    I &= I_{\mathrm{EH}} + I_{\mathrm{GHY}} + I_{\mathrm{brane}}, \\
    I_{\mathrm{EH}} &= - \frac{1}{16 \pi G_N} \int_M \dd[3]{x}  \sqrt{g} (R - 2 \Lambda), \\
    I_{\mathrm{GHY}} &= - \frac{1}{8 \pi G_N} \int_{\Sigma} \dd[2]{x} \sqrt{h} K, \\
    I_{\mathrm{brane}} &= - \frac{1}{8 \pi G_N} \int_Q \dd[2]{x} \sqrt{h} (K - h^{ab} \partial_a \phi \partial_b \phi ),   \label{onacn}
  \end{aligned}
\end{equation}
where $g_{\mu \nu}$ is the bulk metric, $h_{ab}$ is the induced metric on $\Sigma$ or $Q$, and $K_{ab}$ is the extrinsic curvature. 
$\Lambda$ is a negative cosmological constant, written by using AdS radius $l$ as $\Lambda = -l^{-2}$. 
Here and hereafter we set $l$ = 1. 

Though we impose Dirichlet boundary condition $\delta h_{ab}=0$ on asymptotic AdS boundary $\Sigma$,  we choose the Neumann one on the EOW brane $Q$: 
\begin{equation}
  K_{ab} - K h_{ab} + (h^{cd} \partial_c \phi \partial_d \phi) h_{ab} - 2 \partial_a \phi \partial_b \phi = 0. \label{neumann}
\end{equation}
This equation governs the dynamics of EOW brane. Note that 
the equation of motion for the scalar field $\phi$ follows from a linear combination of (\ref{neumann}). 

At finite temperature, as a solution to the three dimensional Einstein equation, there are two well known geometries: thermal AdS (TAdS) and BTZ black hole, whose metrics read
\ba
&& \mbox{TAdS:}\ \ \ \    ds^2 = \frac{d \tau^2}{z^2} + \frac{d z^2}{h(z) z^2} + \frac{h(z)}{z^2} dx^2, \label{TAdSM} \\
&& \mbox{BTZ:}\ \ \ \  ds^2 = \frac{h(z)}{z^2} d \tau^2 + \frac{d z^2}{h(z) z^2} + \frac{dx^2}{z^2},\no
&& \mbox{where we set}\ \ \ \ h(z) = 1 - \frac{z^2}{a^2}. \label{BTZM}
\ea
In the TAdS,  the $x$ direction must have periodicity $2 \pi a$ to avoid a conical singularity at $z = a$. In the BTZ, we have to compactify $\tau$ direction as $\tau \sim \tau + 2 \pi a$. 

They undergo a phase transition at some temperature \cite{Hawking:1982dh}. A similar phase transition occurs in the AdS/BCFT \cite{Takayanagi:2011zk, Fujita:2011fp}. In the presence of a scalar field localized on the EOW brane, the gravity dual has a richer structure in that there is a phase transition from the confined (thermal AdS) phase to the deconfined (BTZ) phase as we increase the difference $\Delta \phi$ of the scalar field between the two boundaries $x=x_*$ and $x=-x_*$ as found in \cite{Kanda:2023jyi,Kanda:2023jyi}. For our purpose we assume the scalar field $\phi$ is imaginary valued to satisfy the non-Hermitian boundary condition (\ref{nhbc}). Thus we introduce the scalar field $\vp$ via 
\be
\phi=i\vp,  \label{imags}
\ee
which takes real values. As we will explain later, the BCFT lives on a single interval given by the union of $-\pi a\leq x\leq -x_*$ and $x_*\leq x\leq \pi a$ (note that we identify $x=\pi a$
with $x=-\pi a$). Therefore its width is given by 
\ba
\Delta x=2\pi a-2x_*. \label{delatxd}
\ea
In our PT-invariant BCFT, owing to the conformal invariance, independent parameters which control the phases are the following two:
\begin{equation}
  \frac{\Delta x}{\beta} , \qquad \Delta \vp := \vp(x_*) - \vp(- x_*)=2\lambda.
\end{equation}


\subsection{PT-invariant Thermal AdS Phase $0\leq \Delta \vp \leq \vp_c$} 
In the thermal AdS$_3$ (\ref{TAdSM}), 
the equation of motion in AdS/BCFT for an EOW brane with the localized scalar, given by (\ref{neumann}), reads (we write the derivative w.r.t $x$ as $\frac{dz}{dx}=\dot{z}$)
\ba
&& \eta\cdot \frac{2h^3-zh^2h'-3zh'\dot{z}^2+2h\dot{z}^2+2hz\ddot{z}}{2z^2h\left(h+\frac{\dot{z}^2}{h}\right)^{\frac{3}{2}}}=-\frac{h}{\dot{z}^2+h^2}\dot{\phi}^2 \no 
&& \hspace{174pt} = \frac{h}{\dot{z}^2+h^2}\dot{\varphi}^2,\no 
&& \eta\cdot \frac{h\s{h+\frac{\dot{z}^2}{h}}}{z^2}=\dot{\phi}^2=-\dot{\vp}^2,
\label{eomchi}
\ea
where $\eta=\pm$ corresponds to the orientation of the EOW brane such that the gravity dual region is in the smaller $z$ side for $\eta=1$ and in the larger side for $\eta=-1$. To make $\phi$ real valued, $\eta=1$ was chosen in \cite{Kanda:2023zse}, where the gravity dual looks like the left panel of Fig.\ref{fig:TAdSEOW}. Instead, for our purpose of non-Hermitian deformations, we assume $\phi$ is imaginary by setting (\ref{imags}). Thus we choose $\eta=-1$ (i.e. the right panel of Fig.\ref{fig:TAdSEOW}), which explains (\ref{delatxd}).

\begin{figure}
    \centering
    \includegraphics[width=4cm]{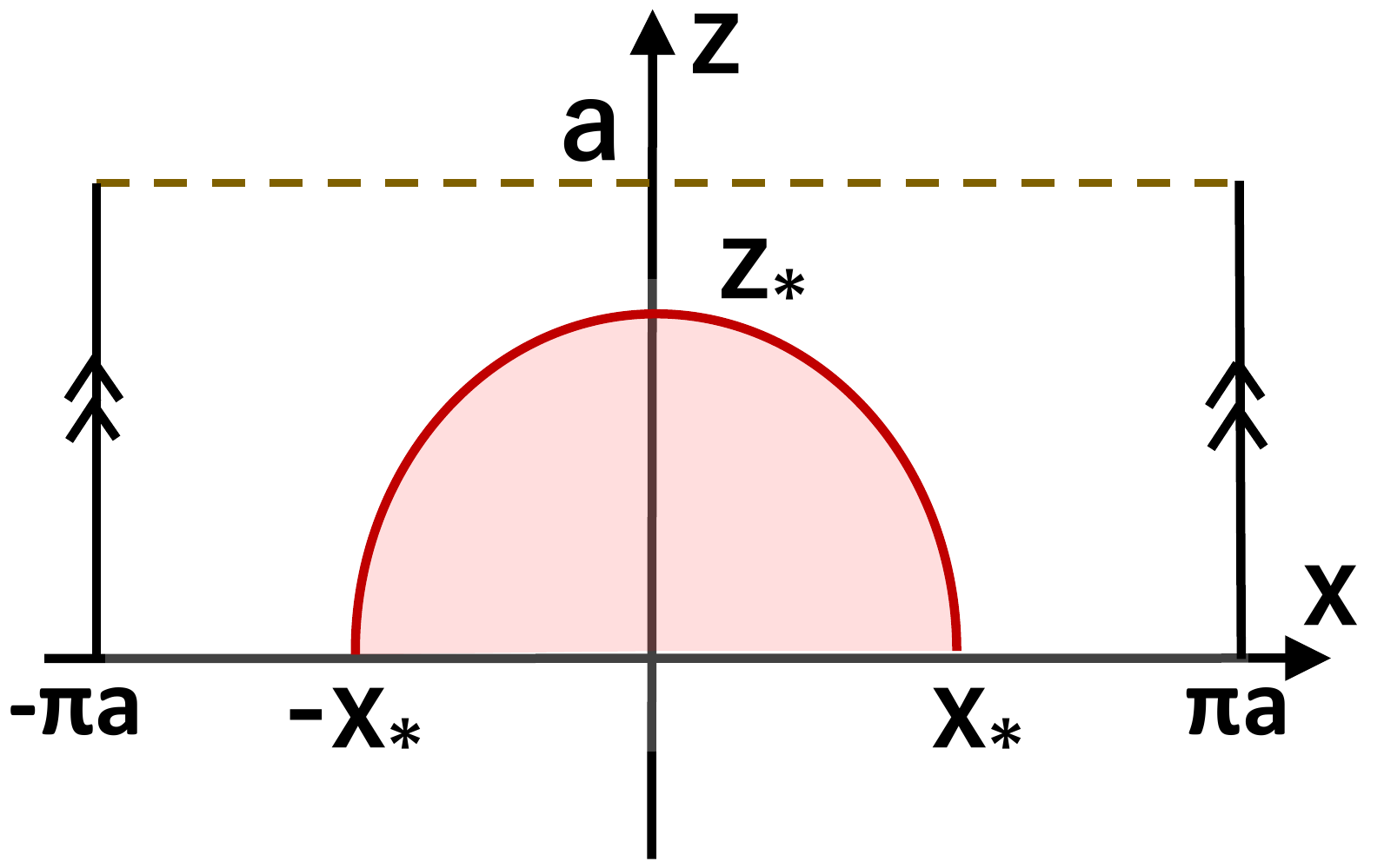}
        \includegraphics[width=4cm]{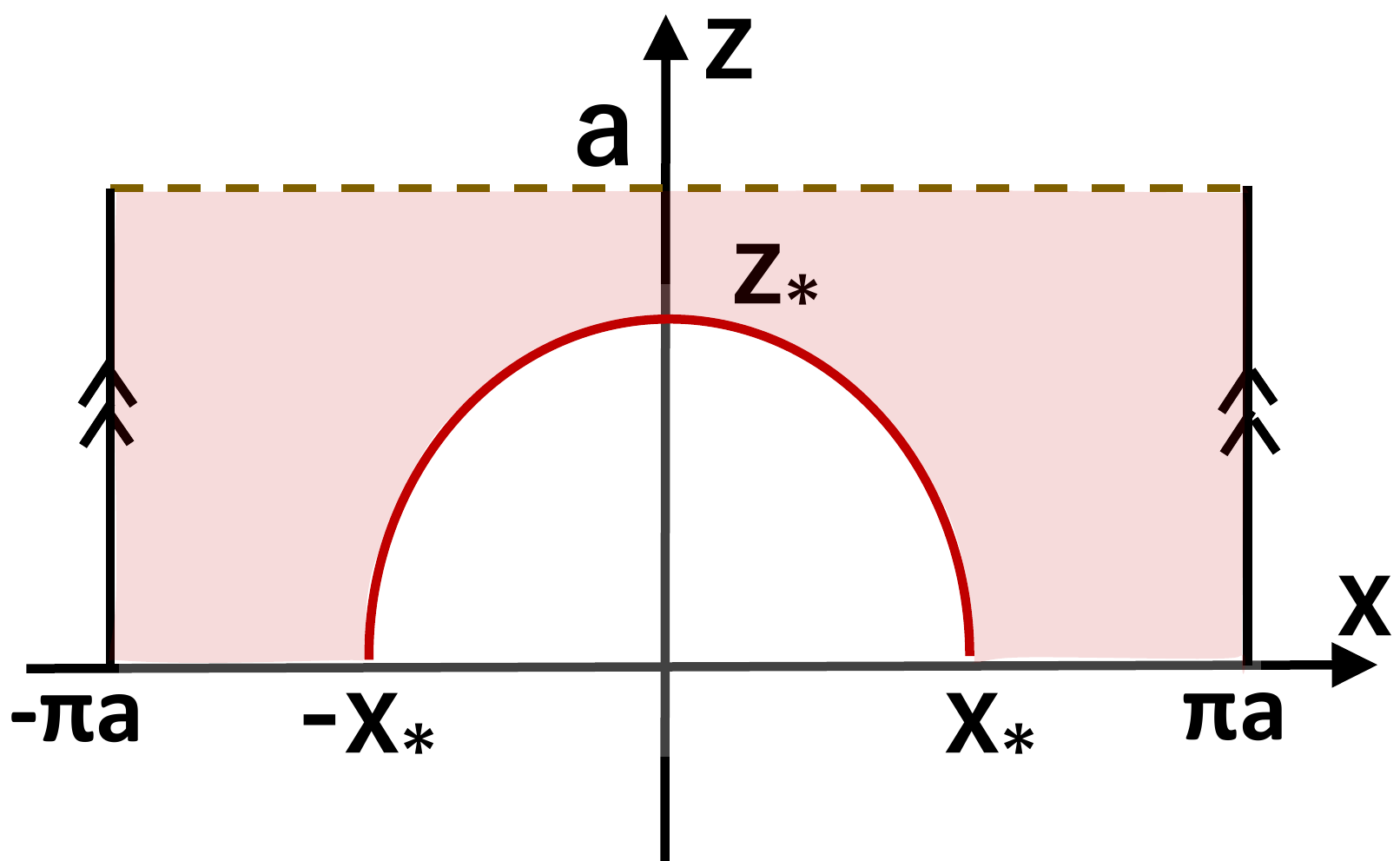}
    \caption{Sketches of gravity duals of the BCFT with the asymmetric exactly marginal deformation (colored region), where the red curve describes the EOW brane. The left and right panels correspond to the case where $\phi$ is real and imaginary, respectively. }
    \label{fig:TAdSEOW}
\end{figure}

By introducing the potential (identical to the one in \cite{Kanda:2023zse})
\ba
U(z)=-\frac{9a^2z^2}{2}+\frac{9a^6}{2(a^2-z^2)},
\label{pote}
\ea
we can rewrite (\ref{eomchi}) for $\eta=-1$ with (\ref{imags}) as follows
\ba
&& \dot{z}=\frac{h^{\frac{3}{2}}\s{2U(z_*)-2U(z)}}{3z^2},\no
&& \dot{\vp}=\frac{h^{\frac{3}{4}}}{z}\left(1+\frac{2h}{9z^4}(U(z_*)-U(z))\right)^\frac{1}{4},
\ea
where $z_*$ is the turning point of the brane such that 
$x(z_*)=0$. The solution is parameterized by $z_*$ and the physically important values $x_*$ and $\Delta\vp$ are given by a function of its dimensionless ratio $s=\frac{z_*}{a}$ as
\ba
\frac{2x_*}{a}=X_i(s),
\ \ \ \ \Delta \vp=\Phi_i(s), \ \ (i=T, B),\label{txph}
\ea
where $i=T$ or $B$ denotes the thermal AdS or BTZ phase, respectively. In the thermal AdS phase, they are explicitly expressed as
\begin{align}
 X_T(s) &= \frac{2\sqrt{1-s^2}}{s}\left[\Pi\!\left(s^2,s^2-1\right)-K\!\left(s^2-1\right)\right],  \label{xts} \\
  \Phi_T(s) &= 2(1-s^2)^{\frac{1}{4}} K\qty(-(1-s^2)),
  \label{pts}
\end{align}
where $K(x)$ and $\Pi(x,y)$ are the complete elliptic integrals of the first and the third kind, respectively.

When the deformation is small i.e. $\Delta \vp\ll 1$ (or $s\simeq 1$), we find approximately
$X_T(s)\simeq \pi- \pi \s{2(1-s)}$ and $\Phi_T(s)\simeq 
2^{\frac{1}{4}}\pi (1-s)^{\frac{1}{4}}$, leading to 
\ba
X_T(s)\simeq \pi-\frac{\Phi_T(s)^2}{\pi}.
\ea

On the other hand, there is an upper bound 
$\Delta \vp\leq \vp_c$ as long as we assume $z_*$ is real, given by
\ba
\vp_c=2K(-1)=\frac{2\s{\pi}\Gamma\left(5/4\right)}{\Gamma\left(3/4\right)}\simeq 2.622.
\ea
When we approach this upper bound, we obtain
$X_T(s)\simeq \frac{2\s{\pi}\Gamma(7/4)}{3\Gamma(5/4)}s$ and $\Phi_T(s)\simeq \vp_c-
\frac{\s{\pi}\Gamma(3/4)}{2\Gamma(1/4)}s^2$, leading to 
\ba
\Phi_T(s)\simeq \vp_c-
\frac{\ap}{4}X_T(s)^2,
\ea
where we defined $\ap=\frac{\Gamma(1/4)}{2\s{\pi}\Gamma(3/4)}$.

As explained in the Appendix \ref{sec:action}, the on-shell action is given by
\begin{equation}
 I_T = - \frac{1}{16 \pi G_N} \cdot \frac{\beta}{\Delta x} \left( 2 \pi -  X_T [\Phi_T^{-1} (\Delta \vp)] \right)^2 \label{tadaction}
\end{equation}

In terms of energy of thermal AdS$_3$ phase in the low temperature limit, we find 
\ba
E_T=\lim_{\beta\to\infty}\frac{I_{T}}{\beta}=- \frac{c}{24 \pi\Delta x} \left( 2 \pi -  X_T [\Phi_T^{-1} (\Delta \vp)] \right)^2, \label{Teneg}
\ea
where $c$ is the central charge of the dual CFT.
This is plotted in Fig.\ref{fig:energya}.
When $\Delta\vp\ll 1$, we obtain 
$E_T\simeq -\frac{\pi c}{24\Delta x}\left(1+\frac{2}{\pi^2}(\Delta \vp)^2\right)$. On the other hand, in the limit $\Delta\vp\to \vp_c$, we find
\ba
E_T &\simeq &-\frac{\pi c}{6\Delta x}+\frac{c}{3\s{\ap}\cdot\Delta x}\s{\vp_c-\Delta\vp}. \label{chilimg}
\ea
This implies that beyond the point $\Delta \vp=\vp_c$, the energy becomes complex valued.

We conclude this subsection by commenting that in the PT invariant phase, there is another solution of Einstein equation which can be obtained by replacing $z_*$ with $- z_*$. 
This solution actually has lower free energy than the one with positive $z_*$, and coincides with it at $\Delta \vp=\vp_c$. 
However, we discarded the negative $z_*$ solution because when one takes $\Delta \vp = 0$, it gives the energy $E_T^{(z_* < 0)} = - \frac{9 \pi c}{24 \Delta x}$.
Thus it violates the unitarity without non-Hermitian deformation, which is not what we want.
This implies that $\Delta \vp=\vp_c$ is not an exceptional point at least in the standard sense, in spite of its square root behavior. 

\begin{figure}
    \centering
    \includegraphics[width=8cm]{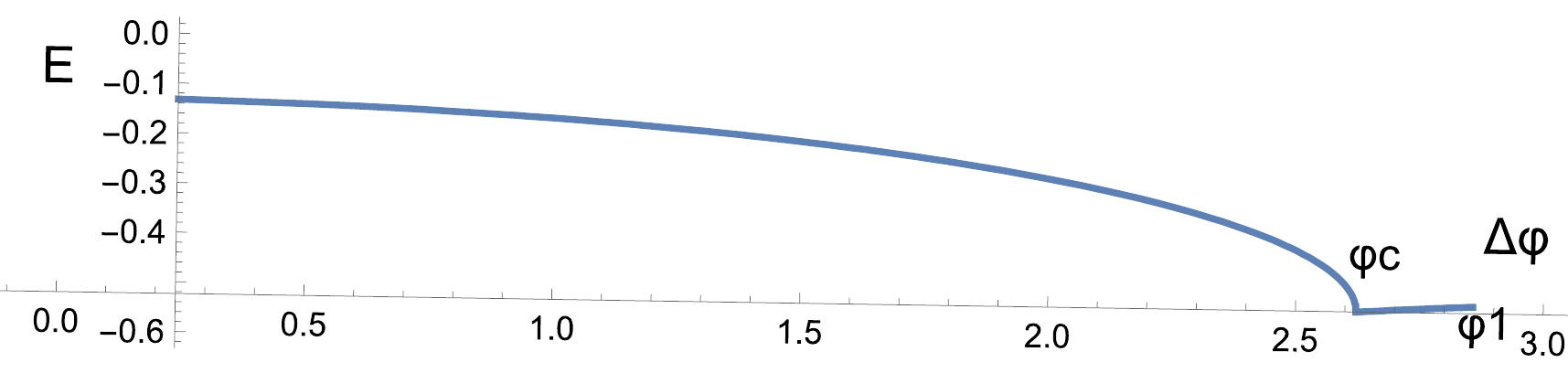}
    \caption{Plots of the energy $E_T$ (\ref{Teneg}) as a function of $\Delta \varphi$ for $0\leq \Delta\vp\leq \vp_1$. For $\vp_c<\Delta \vp\leq \vp_1$, where $E$ becomes complex, we plotted the real part of $E$. We set $\Delta x=1$ and $c=1$.}
    \label{fig:energya}
\end{figure}


\subsection{ PT-violating TAdS Phase $\vp_c<\Delta\vp\leq \vp_1$} 
As we have seen, there is an upper bound $\Delta \vp\leq \vp_c$ for the connected EOW brane solution. Our interpretation as a PT-symmetric system implies that we will get into the spontaneous breaking of the PT-symmetry for 
$\Delta \vp> \vp_c$. Indeed for this range, the solution becomes complex valued as we will explain below. For this, let us allow $z_*$ to take imaginary values. 

In order to have a real valued coordinate in the BCFT, which corresponds to the AdS boundary $z\to 0$, we introduce the real valued coordinate $(\ti{\tau},\ti{x},\ti{z})$ by the rotation
\ba
(\tau,x,z)=e^{i\theta}(\ti{\tau},\ti{x},\ti{z}), \label{rotatis}
\ea
keeping $a$ real valued, such that in the boundary limit $z\to 0$ we have $ds^2\simeq \ti{z}^{-2}(d\ti{\tau}^2+d\ti{z}^2+d\ti{x}^2)$. We identify $(\ti{\tau},\ti{x})$ at $\ti{z}=0$ with the coordinate of the BCFT. We choose the UV cut off by $\ti{z}=\ep$, which is a real infinitesimally small constant. The bulk metric is given by (\ref{TAdSM}) with the coordinate transformation (\ref{rotatis}) performed. Since $z_*$ and $s$ take the imaginary values, $X_T(s)$ in (\ref{xts}) becomes pure imaginary, though we keep $\Phi_T(s)$ in (\ref{pts}) real valued. At $\Delta \vp=\Phi_T=\vp_c$, we have $s=0$. As $|s|$ gets larger,  $\Phi_T(s)$ increases until it reaches its maximum $\vp_1\simeq 2.858$ at $s\simeq 1.937i$. For larger $|s|$, it starts decreasing. Thus this complex solution only covers $\vp_c<\Delta \vp<\vp_1$.

The BCFT extends in the strip region given by $0\leq \ti{\tau}<\beta$ and $-\frac{\Delta x}{2}\leq \ti{x}\leq \frac{\Delta x}{2}$, where both $\beta$ and $\Delta x$ are real valued. In terms of the original space coordinate $x$, the BCFT lives on a tilted interval which is the union of two intervals $\left[-\pi a,-i|x_*|\right]$ and $\left[i|x_*|,\pi a\right]$. Note that we identify $x=-\pi a$ with $x=\pi a$ due to the periodicity. This leads to the phase angle 
\ba
e^{i\theta}=\frac{2\pi-i|X_T(s)|}{\s{4\pi^2+|X_T(s)|^2}}.
\ea
Note that $\Delta x$ is now defined as the absolute value i.e. $\Delta x/a=\s{4\pi^2+|X_T(s)|^2}$. 

From the metric after the coordinate transformation (\ref{rotatis}), the holographic energy stress tensor gets multiplied by the phase factor $e^{2i\theta}$.
In the end, the energy $E_T$ in this PT violating phase, turns out to be still given by the previous formula (\ref{Teneg}). However, in this imaginary $s$ solution, the energy becomes complex valued. This means that the PT-symmetry gets spontaneously broken. Also notice that the configuration with the complex conjugate energy $E^*_T$ can be found simply by setting $i|z_*|\to -i|z_*|$.
The real part of $E_T$ is plotted in Fig.\ref{fig:energya}.

\subsection{PT-invariant BTZ Phase $0\leq \Delta\vp\leq\vp_1$} 
Now let us move on to the BTZ phase (\ref{BTZM}). The Neumann boundary condition (\ref{neumann}) for a connected EOW brane leads to
\ba
&& \dot{z}=\frac{h(z)\s{2U(z_*)-2U(z)}}{3z^2},\no
&& \dot{\vp}=\frac{1}{zh(z)^{\frac{1}{4}}}\left(1+\frac{2h}{9z^4}(U(z_*)-U(z))\right)^\frac{1}{4},
\ea
where the potential $U(z)$ is identical to (\ref{pote}).
By integrating these differential equations, we obtain 
$X_B(s)$ and $\Phi_B(s)$ defined in (\ref{txph}) as follows:
\begin{equation}
    \begin{aligned}
    X_B(s) &\!=\! \frac{2s}{\sqrt{1-s^2}\sqrt{2-s^2}}\!\left[\!\Pi\!\!\left(\frac{1-s^2}{2-s^2}\!,\!\frac{1}{2-s^2}\right)\!-\!K\!\left(\!\frac{1}{2-s^2}\!\right)\!\right],\\
    \Phi_B(s) &= \frac{2(1-s^2)^{\frac{1}{4}}}{\sqrt{2-s^2}}K \qty(\frac{1}{2-s^2}). 
    \end{aligned}
    \label{solbtz}
\end{equation}
We find $X_B(0)=X_B(1)=0$, $\Phi_B(0)=\vp_c$ and $\Phi_B(1)=0$. Moreover, both $X_B$ and $\Phi_B$ take their maximum values $X_B\simeq 1.002$ and $\Phi_B\simeq 2.858$ at the same point $s\simeq 0.889$. 
We note that for $\vp_c<\vp<\vp_1$, there is a second candidate solution for BTZ with connected brane (see Appendix \ref{ap:btzinv}). However it always has larger energy than the one we considered above, so that we ignore it.

Furthermore, in the BTZ phase, we also have the disconnected configuration of EOW brane, consisting of two parallel branes each at $x=-x_*$ and $x=x_*$. We write the on-shell action of the connected EOW brane and the disconnected ones as $I_B$ and $I_D$, respectively.
As computed in the Appendix \ref{sec:action}, they are explicitly given by
\begin{equation}
  \begin{aligned}
      & \mbox{Connected:}\ \  I_B = - \frac{\pi\Delta x}{4 G_N\beta}- \frac{1}{4 G_N} X_B \left[ \Phi_B^{-1} (\Delta \varphi) \right], \\
      & \mbox{Disconnected:}\ \  I_D = - \frac{\pi \Delta x}{4 G_N \beta}.
  \end{aligned} \label{btzaction}
\end{equation}
Since $X_B$ is positive, we find that the connected EOW brane is always favored under our non-Hermitian deformation.

\subsection{Phases for $\Delta\vp>\vp_1$} We can consider the solution for even larger $\Delta \varphi$ by allowing $s$ to take general complex values. 
In this region, not only TAdS but BTZ solutions have complex metrics, which we again interpret as the transition to PT-broken phase. 
The detailed computation is shown in the Appendix \ref{ap: solution}, and here we just indicate the basic idea and final results. 

We can write down $X_{T/B}$ and $\Phi_{T/B}$ explicitly by using some elliptic integrals \eqref{xts}, \eqref{pts} \eqref{solbtz} (see also the Appendix \ref{ap:analytical}), so that we can easily extend the domain of $s$ to the full complex plane at least  numerically. 
Since we are considering real $\Delta \varphi$ setup, we restrict the domain on which $\Phi_{T/B}$ becomes real. 
This procedure gives some contours in the complex $s$ plane. 
Evaluating $\Phi_{T/B}$ along that curve shows that it indeed becomes real number, greater than $\varphi_1$. 
In the same way, evaluating $X_{T/B}$ along the curve and substituting it into on-shell action \eqref{tadaction}, \eqref{btzaction}, we can get free energy in this regime. 
Note that we also keep $\Delta x / \beta$ to be a real number. 

Same as the $\Delta \varphi < \varphi_1$ case, we have three phases: TAdS, BTZ with connected EOW brane, and BTZ with disconnected EOW brane. 
By minimizing the real part of free energy, we discuss the most favored phase for given $\Delta x / \beta$ and $\Delta \varphi$. 
Note that disconnected BTZ solution always has larger free energy than that of connected BTZ, so it is never dominant in our setup. 

The phase transition line between TAdS and connected BTZ is shown in Fig.\ref{fig:phased}. 
The phase diagram includes four phases: TAdS and BTZ, both have PT-invariant phase and PT-broken phase.

\section{Novel Quantum Quenches from BCFT} 

Consider the thermal AdS phase of our setup with the imaginary scalar in the range $0\leq\Delta \vp\leq\vp_c$. We now perform the Wick rotation $\tau'=x$ and $x'=\tau$. At the time $\tau'=-\frac{\Delta x}{2}$ the state is given by the boundary state with the exactly marginal perturbation turned on:
\ba
|B(\Delta\vp/2)\lb=\exp\left[i\frac{\Delta\vp}{2}\int^{\infty}_{-\infty} dx'O(x')\right]|B_0\lb,
\ea
where $|B_0\lb$ is the Cardy state \cite{Cardy:1989ir} before the boundary deformation.

\begin{figure}
    \centering
    \includegraphics[width=9cm]{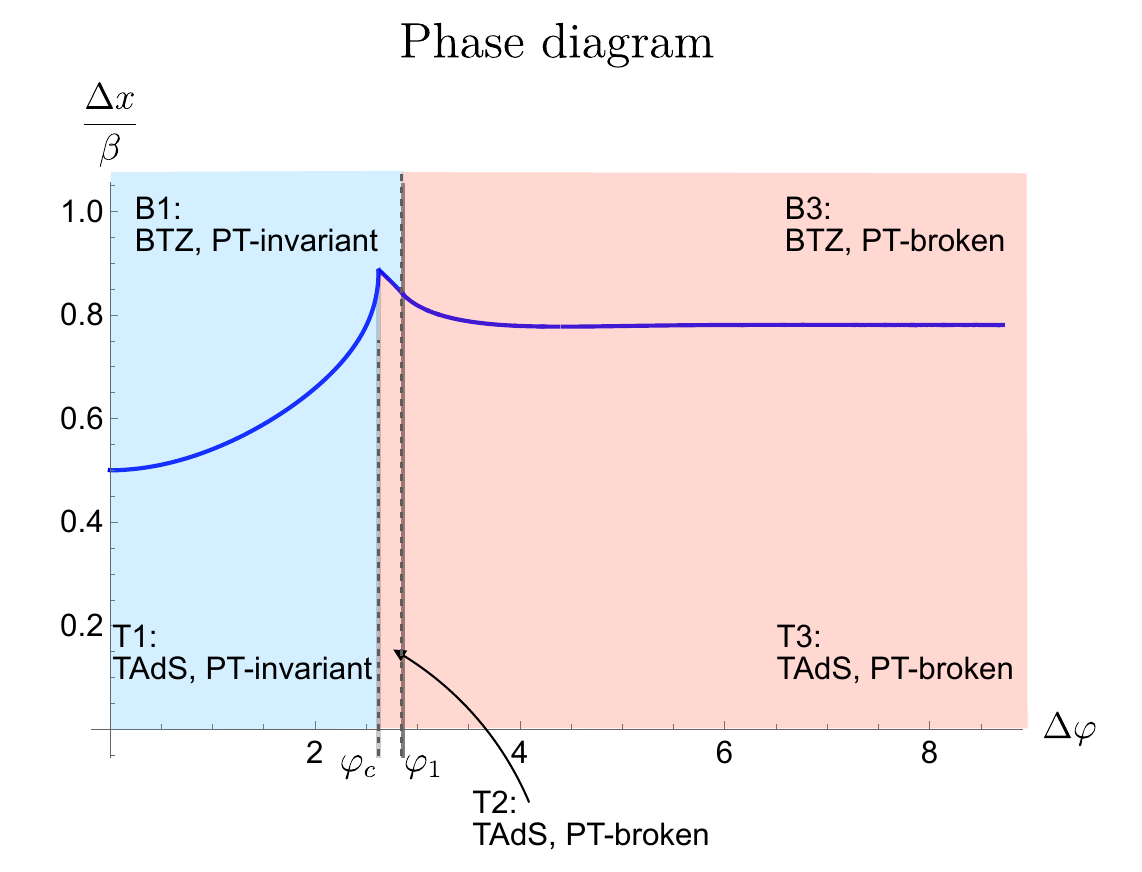}
    \caption{A sketch of phase diagram. The PT-symmetry is spontaneously broken in the red colored phases, while it is unbroken in the blue colored ones.}
    \label{fig:phased}
\end{figure}

Then the partition function on the strip $-\frac{\Delta x}{2}\leq \tau'\leq \frac{\Delta x}{2}$ is written as 
\ba
\la B(\Delta\vp/2)|e^{-\Delta x H}|B(\Delta\vp/2)\lb.
\ea
Here, we assume the standard Hermitian conjugation to obtain $\la B(\Delta\vp/2)|$ and this flips the sign of $\vp$.  Thus our BCFT on the strip provides the genuine density matrix as opposed to the case where $\phi$ is real studied in \cite{Kanda:2023zse,Kanda:2023jyi}, which ends up with a transition matrix \cite{Nakata:2021ubr}. Therefore we can define the entanglement entropy for its time evolution
\ba
|B(t,\Delta\vp/2)\lb=e^{-itH}e^{-\frac{\Delta x}{2}H}|B(\Delta\vp/2)\lb.
\ea
Note that its density matrix is Hermitian and is among regular quantum states in a unitary holographic CFT.

Now we consider the entanglement entropy for the subsystem $A$ given by a half of the total system. Then the time evolution of entanglement entropy 
for our boundary state is dictated by the value of $a$ as
\ba
\Delta S_A(t)\simeq \frac{c}{6a}t,
\ea
at late time. Here $\Delta S_A$ is the difference between the actual entanglement entropy and that for the CFT vacuum. In our Wick rotated setup we can regard the effective temperature $T_{\mbox{eff}}$ as 
\ba
T_{\mbox{eff}}=\frac{1}{2\Delta x}.
\ea
The derivation is essentially the same as the similar setup for the real valued $\phi$ studied in \cite{Kanda:2026jyk}, whose density matrices become non-Hermitian as opposed to our case here.

Thus we obtain the late time behavior:
\ba
\Delta S_A(t)\simeq \frac{c}{3}\cdot T_{\mbox{eff}}\cdot \left(2\pi-X_T[\Phi^{-1}_T(\Delta\vp)]\right)\cdot t.
\ea
At $\Delta\vp=0$ we reproduce the known result $\Delta S_A(t)\simeq \frac{\pi c}{3\beta}t$ \cite{Calabrese:2005in,Hartman:2013qma}. However, the gradient of the linear growth increases as $\Delta\vp$ gets larger. It reaches the maximum at $\Delta\vp=\vp_c$, where we find the double growth: $\Delta S_A(t)\simeq \frac{2\pi c}{3\beta}t$. Even though the original argument in \cite{Calabrese:2005in,Hartman:2013qma} using the replica method and conformal map looks universal, this assumes that the upper and lower boundaries at $\tau'=\pm \frac{\Delta x}{2}$ map into a single boundary of the upper half plane. In general for $\Delta\vp\neq 0$, we need to insert a boundary condition changing operator as they have the marginal deformations with the different signs: $\pm i\frac{\Delta\vp}{2}\int dx' O(x')$. 
See \cite{Geng:2021iyq} for a reference dealing with the replica trick and the insertion of boundary condition changing operator. 
Refer also to \cite{Hoback:2026yqj} for another realization of such a generalized global quench using de Sitter branes and to \cite{Tang:2026ccf} for related setups in free fermion/boson CFTs as well as lattice models. See also  \cite{Kusuki:2023ckn} for the enhancement of effective temperature for generic quantum states. 

\section{Discussions} 
In this article, we presented an explicit and simple holographic dual of a two dimensional conformal field theory with PT-symmetric boundary conditions, by applying the AdS/BCFT duality. We showed that as the strength of the non-Hermitian PT-symmetric interactions parametrized by $\Delta \vp$ increases, our system experiences a spontaneous PT-symmetry breaking. Unlike the usual PT-symmetric system, however, we should not regard the transition point as an exceptional point, where two energy eigenstates should coincide. In our model, there is another candidate solution of Einstein equation which merges to our solution at the transition point, but we discarded it based on physical intuition.  We also considered its Wick rotated setup as new quantum quenched states and show that its growth of entanglement entropy can be larger than the standard results obtained from usual Cardy states. Clearly, our model can be generalized in many different directions: higher dimensional setup, or adding a potential term in the scalar field. The analysis of their phase structures will be an interesting future problem.

\vspace{5mm}
\section*{Acknowledgments}
We are grateful to Hiroki Kanda, Taishi Kawamoto, Shiraz Minwalla, Robert Myers, Shinsei Ryu, Kazuaki Takasan, Seiji Terashima and Zixia Wei for useful discussions. This work is supported by MEXT KAKENHI Grant-in-Aid for Transformative Research Areas (A) through the ``Extreme Universe'' collaboration: Grant Number 21H05187. TT is also supported by JSPS Grant-in-Aid for Scientific Research (B) No.~25K01000. 
RM is supported by Grant-in-Aid for JSPS Fellows No.~26KJ1548. 
NN is supported by Grant-in-Aid for JSPS Fellows No.~26KJ1546.

\appendix

\section{Simple example of PT-invariant BCFT}\label{ap:free}

Consider a $c=1$ free massless scalar CFT in two dimensions. We write the scalar field as $X(\tau,x)$. We introduce the boundaries at $x=0$ and $x=\pi$. We consider the Dirichlet boundary condition, which preserves the conformal invariance, as follows:

\def \TL {\text{L}}
\def \TR {\text{R}}
\begin{align}
  X(\tau,0) &= x_0+ i\lambda \nonumber\\
  X(\tau,\pi) &= x_0 - i\lambda.
\end{align}
When $\lambda$ takes real values, this boundary condition becomes non-Hermitian and PT invariant because the changes of the values of $X$ at the boundaries $x=0,\pi$ are described by the exactly marginal boundary perturbation of the form \eqref{BCFTaction}. 

Now let us examine the spectrum of this BCFT. The standard mode expansion reads
\begin{align}
  X(\tau,x) &= x_0 +i\lambda - 2i\lambda \frac{x}{\pi} +2\sum_{n\neq 0}\frac{a_n}{n}e^{-n\tau}\sin nx,
\end{align}
where $a_n$ is quantized by the commutation relation $[a_n,a_m]=n\delta_{n+m,0}$. By choosing the vacuum by $a_n|0\lb=0$ for $n>0$ as usual, the energy for each excitation is found to be (we remove the Casimir energy part)
\begin{align}
  H &= -\frac{\lambda^2}{2\pi^2} +\sum^\infty_{n=1} a_{-n}a_n.
\end{align}
This shows that due to the non-Hermitian boundary condition, the energy can be negative, though it is bounded from below because $\sum^\infty_{n=1} a_{-n}a_n\geq 0$. 

It is useful to study this setup of BCFT (or open string picture) by performing the Wick rotation $(\tau',x')\equiv(x,\tau)$ i.e. from closed string theory viewpoint. 
We now have the initial state at $\tau'=0$ given by the boundary state \cite{Cardy:1989ir}, expressed as $|B_1\lb$ defined by the Dirichlet boundary condition $X=x_0+i\lambda$. Then we have a Euclidean time evolution until $\tau'=\pi$. Finally it is projected to the final state given by the other boundary state defined by the Dirichlet boundary condition $X=x_0-i\lambda$. The left and right-moving mode of the scalar field $X=X_{\TL}+X_{\TR}$ is expanded as follows:
\begin{align}
  X_{\TL} &= x_{(L)} - p_{(L)}(i\tau'+x') + i\sum_{n\neq 0} \frac{\ap_n}{n} e^{in(i\tau'+x')} \\
X_{\TR} &= x_{(R)} -p_{(R)}(i\tau'-x')  + i\sum_{n\neq 0} \frac{\tilde{\ap}_n}{n} e^{in(i\tau'-x')},
\end{align}
where the canonical quantization leads to $[x_{(L)},p_{(L)}]=[x_{(R)},p_{(R)}]=i$ and $[\ap_n,\ap_m]=[\ti{\ap}_n,\ti{\ap}_m]=n\delta_{n+m,0}$. When we compactify the scalar on a circle with the radius $R$, the zero modes are quantized as
$p_{(L)}=\frac{k}{R}+\frac{wR}{2}$ and $p_{(R)}=\frac{k}{R}-\frac{wR}{2}$, where $k$ and $w$ are integers.

At $\tau'=0$, the quantum state is given by the boundary state $\ket{B_1}$ for our Dirichlet condition:
\begin{align}
  \ket{B_1} &\propto \sum_{k\in\mathbb{Z}} e^{i(x_0+i\lambda)\frac{k}{R}} e^{\sum^\infty_{n=1} \frac{1}{n} \ap_{-n}\tilde{\ap}_{-n}}|k,w=0\lb. 
\end{align}
Its Hermitian conjugation is given by
\begin{align}
  \bra{B_1} &\propto \la k,w=0|\sum_{k\in\mathbb{Z}} e^{-i(x_0-i\lambda)\frac{k}{R}} e^{\sum^\infty_{n=1} \frac{1}{n} \ap_{n}\tilde{\ap}_{n}}.
\end{align}
The partition function for this free scalar field theory on the interval $0\leq \tau'\leq \pi$ is given by 
$\la B_1|e^{-\pi H}\ket{B_1}$, which is well-defined in spite of the non-Hermitian boundary condition. Also it is clear that the density matrix $\rho=e^{-\frac{\pi}{2}H}|B_1\lb\la B_1|e^{-\frac{\pi}{2}H}$ is Hermitian. Note that as opposed to our current setup, if we consider the real valued deformation (i.e. $i\lambda$ is real valued), which is Hermitian, $|B_1\lb$ is no longer Hermitian conjugate to $\la B_1|$ and thus $\rho$ is not Hermitian. Such non-Hermitian density matrices (or called transition matrices) were discussed in the context of BCFT in \cite{Kanda:2023zse,Kanda:2023jyi,Kanda:2026jyk}, where interesting phase transitions in pseudo entropy \cite{Nakata:2021ubr} were found.

\section{Evaluation of On-Shell Action}\label{sec:action}

\subsection{Thermal AdS}
In the thermal AdS solution, since $x$ direction has the periodicity $2 \pi a$, the width $\Delta x$ is given by $\Delta x = 2 \pi a - 2 x_*$. 
Note that $2 x_* / a$ was in the range $[0, \pi]$.
Thus in the current regime, the width of BCFT must satisfy $\pi \leq \Delta x / a \leq 2 \pi$. 

Now let us evaluate the free energy at a fixed inverse temperature $\beta$, that is given by the on-shell action Eq.\eqref{onacn}. 
The bulk geometry satisfies 
\begin{equation}
  R = -6, \qquad K|_{\Sigma} = 2 + O(\epsilon^4). 
\end{equation}
Here we set the cutoff $z \geq \epsilon$, to remove the divergence of on-shell action. 
Moreover, we get $K = 0$ on the EOW brane. In addition, to remove the divergence at asymptotic boundary, we have to subtract local counter term:
\begin{equation}
  \begin{aligned}
      I_{\mathrm{c.t.}} &= - \frac{1}{16 \pi G_N} \int \sqrt{g^{(0)}} (R^{(0)} + 2) - \frac{1}{8 \pi G_N} \int_{\Sigma} \sqrt{h^{(0)}} K^{(0)} \\
      &= \frac{\beta \Delta \bar{x}}{4 \pi G_N} \int_{\epsilon}^{\infty} \frac{\dd{z}}{z^3} - \frac{\beta \Delta \bar{x}}{8 \pi G_N} \frac{2}{\epsilon^2}, 
  \end{aligned}
\end{equation}
where $\Delta \bar{x} = \sqrt{h(\epsilon)} \Delta x \simeq \left(1 - \frac{\epsilon^2}{2 a^2} \right) \Delta x$. 
The boundary term on $\Sigma$ cancels as well, and total action reads
\begin{equation}
  \begin{aligned}
      I_T &= \frac{\beta}{4 \pi G_N} \int_{M} \frac{\dd{z} \dd{x}}{z^3} 
      - \frac{\beta \Delta \bar{x}}{4 \pi G_N} \int_{\epsilon}^{\infty} \frac{\dd{z}}{z^3} \\
        & \hspace{20pt} + \frac{\beta}{8 \pi G_N} \int_Q \dd{x} \sqrt{\frac{h}{h^2 + \dot{z}^2}} \dot{\phi}^2 \\
      &= \frac{\beta}{4 \pi G_N} \Bigg[ \Delta x \left( -\frac{1}{2 a^2} + \frac{1}{2 \epsilon^2} \right) \\ 
        & \hspace{50pt} + 2 \int_0^{x*} \left(- \frac{1}{2 a^2} + \frac{1}{2 z(x)^2}\right) \Bigg] \\ 
        & \hspace{20pt} - \frac{\beta \Delta x}{4 \pi G_N} \left( 1 - \frac{\epsilon^2}{2 a^2} \right) \frac{1}{2 \epsilon^2} - \frac{2 \beta}{8 \pi G_N} \int_0^{z_*} \dd{z} \dv{x}{z} \frac{h}{z^2} \\
      &= - \frac{\beta}{8 \pi G_N} \Bigg[ \frac{1}{a^2} \left( \Delta x + 2 x_* - \frac{\Delta x}{2} \right) \\
        & \hspace{60pt} + \int_0^{z_*} \dd{z} \frac{6}{h(z)^{1/2} \sqrt{2 U(z_*) - 2 U(z)}} \\ 
        & \hspace{70pt} - \int_0^{z_*} \dd{z} \frac{6}{h(z)^{3/2} \sqrt{2 U(z_*) - 2 U(z)}} \Bigg] \\
      &= - \frac{\beta \, \Delta x}{16 \pi G_N a^2} \\
      &= - \frac{1}{16 \pi G_N} \cdot \frac{\beta}{\Delta x} \left( 2 \pi -  X_T [\Phi_T^{-1} (\Delta \varphi)] \right)^2
  \end{aligned}
\end{equation}
In the last line, we used
\begin{equation}
  \frac{\Delta x^2}{a^2} = \left( \frac{2 \pi a - 2 x_*}{a} \right)^2 = \left( 2 \pi - X_T (\Phi_T^{-1} (\Delta \varphi)) \right)^2
\end{equation}

\subsection{BTZ with connected EOW brane}

In the BTZ phase with a connected EOW brane, the evaluation of on-shell action slightly changes. 
Now the bulk region generally includes the horizon $z = a$, so $a$ is no longer an independent parameter but a function of temperature: $a = \beta / 2 \pi$. 
Then it naively seems that we cannot choose $\Delta x / \beta$ and $\Delta \varphi$ independently. 
However, we have one new parameter: the periodicity of $x$. 
We write this as $x \sim x + 2 \pi A$, so that $\Delta x = 2 \pi A - 2 x_*$. 
That is to say, $\Delta x$ and $x_*$ do not depend on each other, so we can freely choose $\Delta x$ (while $x_*$ is still a function of $\Delta \varphi$). 

The counter term is
\begin{equation}
  \begin{aligned}
      I_{\mathrm{c.t.}} &= - \frac{1}{16 \pi G_N} \int \sqrt{g^{(0)}} (R^{(0)} + 2) - \frac{1}{8 \pi G_N} \int_{\Sigma} \sqrt{h^{(0)}} K^{(0)} \\
      &= \frac{\bar{\beta} \Delta x}{4 \pi G_N} \int_{\epsilon}^{\infty} \frac{\dd{z}}{z^3} - \frac{\bar{\beta} \Delta x}{8 \pi G_N} \frac{2}{\epsilon^2}. 
  \end{aligned}
\end{equation}
Here we introduced $\bar{\beta} = \sqrt{h(\epsilon)} \beta \simeq \left( 1 - \frac{\epsilon^2}{2 a^2} \right) \beta$.

The boundary term on $\Sigma$ cancels, and on-shell action becomes
\begin{equation}
  \begin{aligned}
    I_B &= \frac{\beta}{4 \pi G_N} \int_{M} \frac{\dd{z} \dd{x}}{z^3} - \frac{\bar{\beta} \Delta x}{4 \pi G_N} \int_{\epsilon}^{\infty} \frac{\dd{z}}{z^3} \\ 
        & \hspace{20pt}+ \frac{\beta}{8 \pi G_N} \int_Q \dd{x} \frac{h}{\sqrt{h + \dot{z}^2}} \dot{\phi}^2 \\
    &= \frac{\beta}{4 \pi G_N} \Bigg[ \Delta x \left( -\frac{1}{2 a^2} + \frac{1}{2 \epsilon^2} \right) + 2 \int_0^{x*} \left(- \frac{1}{2 a^2} + \frac{1}{2 z(x)^2}\right) \Bigg] \\
        & \hspace{20pt} - \frac{\beta \Delta x}{4 \pi G_N} \left( 1 - \frac{\epsilon^2}{2 a^2} \right) \frac{1}{2 \epsilon^2} \\ 
        & \hspace{20pt}- \frac{2 \beta}{8 \pi G_N} \int_0^{x_*} \frac{h}{\sqrt{h + \dot{z}^2}} \frac{(2 h - z h') \sqrt{h + \dot{z}^2}}{2 z^2 h} \\
    &= - \frac{\beta}{8 \pi G_N a^2} \left( \Delta x + 2 x_* - \frac{\Delta x}{2} \right) \\ 
    & \hspace{20pt} + \frac{\beta}{4 \pi G_N} \left[ \int_0^{x_*} \dd{x} \frac{1}{z(x)^2} - \int_0^{x_*} \dd{x} \frac{1}{z(x)^2}\right] \\
    &= - \frac{\beta \, \Delta x}{16 \pi G_N a^2} - \frac{2 \beta x_*}{8 \pi G_N a^2} \\
    &= - \frac{\pi}{4 G_N} \frac{\Delta x}{\beta} - \frac{1}{4 G_N} X_B \left[ \Phi_B^{-1} (\Delta \varphi) \right]
  \end{aligned}
\end{equation}
In the third equality, we used that $h = 1 - z^2 / a^2$ satisfies $2 h - z h' = 2$, and in the final line we used $a = \beta / 2 \pi$. 

\subsection{BTZ black hole with disconnected EOW brane phase}

Then, let us consider another solution, which has two disconnected EOW branes. 
This is obtained by setting $\phi = \mathrm{const.}$ on each brane, whose difference is identified to $\Delta \phi$.
The EOW branes extend along $x = \mathrm{const.}$, so the solution is the same as the standard AdS/BCFT without any scalar field with vanishing tension \cite{Takayanagi:2011zk,Fujita:2011fp}. 
The on-shell action becomes
\begin{equation}
  \begin{aligned}
      I_D &= \frac{\beta}{4 \pi G_N} \int_{-\Delta x / 2}^{\Delta x / 2} \dd{x} \int_{\epsilon}^{a} \frac{\dd{z}}{z^3} - \frac{\bar{\beta} \Delta x}{4 \pi G_N} \int_{\epsilon}^{\infty} \frac{\dd{z}}{z^3} \\
      &= - \frac{\beta \, \Delta x}{16 \pi G_N a^2} \\
      &= - \frac{\pi \, \Delta x}{4 G_N \beta}. 
  \end{aligned}
\end{equation}
In the last line, we used the relation between black hole temperature and horizon radius: $a = r_h^{-1} = \beta / 2 \pi$.

\subsection{Summary}

We summarize the free energy of three phases. 
\begin{equation}
  \begin{aligned}
      & \mathrm{Thermal ~ AdS} \\
      & \hspace{30pt} I_T = - \frac{1}{16 \pi G_N} \cdot \frac{\beta}{\Delta x} \left( 2 \pi -  X_T [\Phi_T^{-1} (\Delta \varphi)] \right)^2 \\
      & \mathrm{BTZ ~ BH ~ (connected)} \\
      & \hspace{30pt} I_B = - \frac{\pi}{4 G_N} \frac{\Delta x}{\beta} - \frac{1}{4 G_N} X_B \left[ \Phi_B^{-1} (\Delta \varphi) \right] \\
      & \mathrm{BTZ ~ BH ~ (disconnected)} \\
      & \hspace{30pt} I_D = - \frac{\pi}{4 G_N} \frac{\Delta x}{\beta}
  \end{aligned} \label{actions}
\end{equation}
For a given $\Delta x / \beta$ and $\varphi$, the phase with lowest free energy is realized. 
We can see that the first term of $I_B$ is identical to $I_D$. 
Therefore, since the second term of $I_B$ is negative, $I_D$ is always larger than $I_B$. 
This means that the BTZ black hole with disconnected EOW brane phase never becomes dominant in this case, unlike the usual real scalar case \cite{Kanda:2023zse}. 
By comparing $I_T$ and $I_B$, the phase transition occurs when
\begin{equation}
  \frac{\Delta x}{\beta} = \left( \frac{\Delta x}{\beta} \right)_{\mathrm{c}}
\end{equation}
with
\begin{equation}
    \begin{aligned}
        \frac{\pi}{4 G_N} \left( \frac{\Delta x}{\beta} \right)^2_{\mathrm{c}} + \frac{1}{4 G_N} X_B[\Phi^{-1}_B (\Delta \varphi)] \left( \frac{\Delta x}{\beta} \right)_{\mathrm{c}} \\
        - \frac{1}{16 \pi G_N} \left( 2 \pi -  X_T [\Phi_T^{-1} (\Delta \varphi)] \right)^2 = 0, 
    \end{aligned}
\end{equation}
and therefore, 
\begin{equation}
    \begin{aligned}
        \left( \frac{\Delta x}{\beta} \right)_{\mathrm{c}} =& \frac{1}{2 \pi} \Bigg[- X_B[\Phi^{-1}_B(\Delta \varphi)] \\ &+ \sqrt{X_B[\Phi^{-1}_B(\Delta \varphi)]^2 + \left( 2 \pi - X_T [\Phi_T^{-1} (\Delta \varphi)] \right)^2} \Bigg]. 
    \end{aligned}
    \label{critical}
\end{equation}
At low temperature $\frac{\Delta x}{\beta} \leq \left( \frac{\Delta x}{\beta} \right)_{\mathrm{c}}$, thermal AdS phase is realized, while at high temperature $\frac{\Delta x}{\beta} \geq \left( \frac{\Delta x}{\beta} \right)_{\mathrm{c}}$, 
BTZ black hole with connected EOW brane phase appears. 

\section{Analytical Expressions} \label{ap:analytical}
We have the analytical expressions of $\Phi_i(s)$ and $X_i(s)$:
\begin{align}
  \Phi_B(s) &= \frac{2(1-s^2)^{\frac{1}{4}}}{\sqrt{2-s^2}}K \qty(\frac{1}{2-s^2}), \\
  \Phi_T(s) &= 2(1-s^2)^{\frac{1}{4}} K\qty(-(1-s^2)),\\
 X_B(s) &= \frac{2s}{\sqrt{1-s^2}\sqrt{2-s^2}}\!\Biggl[\!\Pi\!\left(\!\frac{1-s^2}{2-s^2}\!,\!\frac{1}{2-s^2}\!\right)\!-\! K\!\left(\frac{1}{2-s^2}\!\right)\Biggr],\\
  X_T(s)&= \frac{2\sqrt{1-s^2}}{s}\left[\Pi\!\left(s^2,s^2-1\right)-K\!\left(s^2-1\right)\!\right], \\
\end{align}

where $s=\frac{z_*}{a}$, and $s\in[0,1]$. If we put
\begin{align}
  p=\frac{1}{2-s^2},\ \   r=\frac{1-s^2}{2-s^2},
\end{align}
we can express them using the identical function:
\begin{align}
&  \Phi_B(s) = \Phi_G(p),\ \    X_B(s) = X_G(p),\ \\
&  \Phi_T(s) = \Phi_G(r), \ \ \   X_T(s) = X_G(r),\ \ 
\end{align}
where 
\ba
&&  \Phi_G(x) = 2(x(1-x))^{1/4} K(x),\no
&&  X_G(x) = 2\sqrt{\frac{x\,|2x-1|}{1-x}}\left[\Pi(1-x,x)-K(x)\right].
\ea
Thus, the maximal values of $\Phi_B(s)$ and $\Phi_T^I(s)$ are the same (Fig.\ref{fig:gf}).
\begin{figure}[h]
  \centering
  \includegraphics[width=0.5\columnwidth]{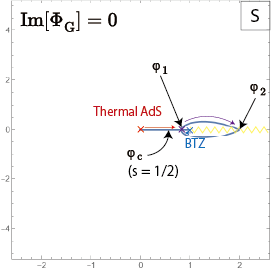}
  \caption{The contour plot of $\Im\Phi_G=0$. 
Starting from $s=1$, the BTZ branch follows the blue path and is analytically continued to the purple branch, while the Thermal AdS branch starts from $s=0$, follows the red path and is continued to the same purple branch. 
Thus $\Phi_T$ and $\Phi_B$ have been analytically continued twice at the same value $\vp_1,\ \vp_2$.}
  \label{fig:gf}
\end{figure}

\section{Details of Solutions in each phase}\label{ap: solution}
Here we present some details of the EOW brane solutions in each phase by showing explicit plots.  

\subsection{Thermal AdS for $0\leq \Delta\vp\leq \vp_1$}
For $0\leq \Delta \vp\leq \vp_c$, $\Delta\vp$ and $x_*$ take real values where $s$ takes values in the range $0\leq s=\frac{z_*}{a}\leq 1$. These are the functions explicitly given by 
Eq.\eqref{xts} and \eqref{pts}, 
and are plotted in Fig.\ref{fig:zxplotTADS}.
We have $s=1$ and $\frac{2x_*}{a}=\pi$ at $\Delta\vp=0$, while we find $s=0$ and $\frac{2x_*}{a}=0$ at $\Delta\vp=\vp_c$. 

\begin{figure}
    \centering
    \includegraphics[width=4cm]{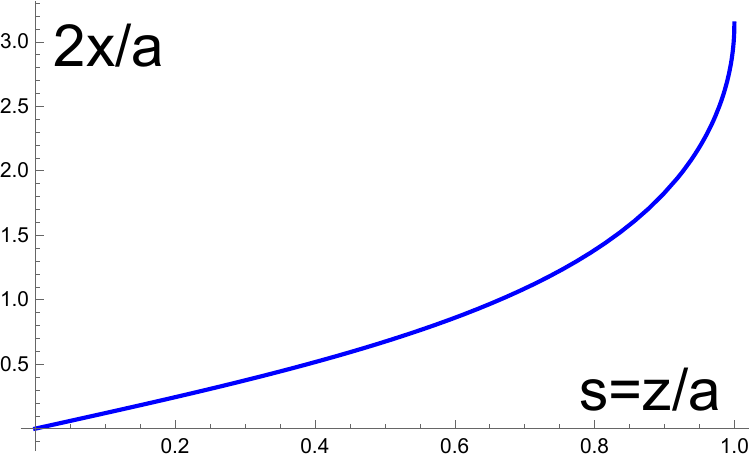}
        \includegraphics[width=4cm]{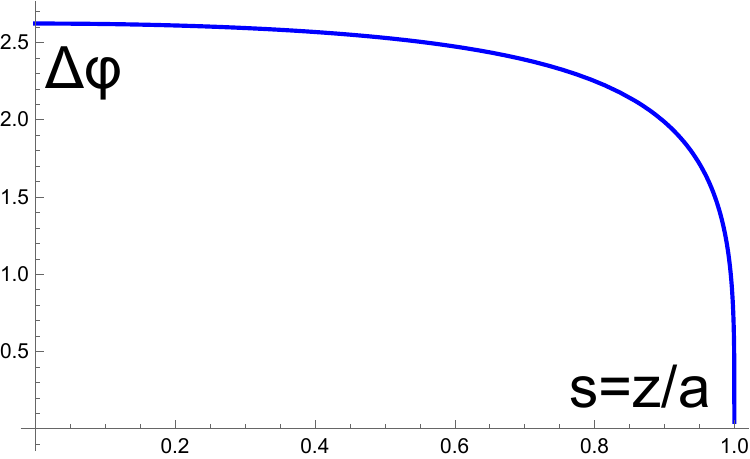}
    \caption{Plots of $\frac{2x_*}{a}=X_T(s)$ (left) and $\Delta \vp=\Phi_T(s)$ as a function of $s$
for $0\leq s\leq 1$.}
    \label{fig:zxplotTADS}
\end{figure}

For $\vp_c\leq \Delta\vp\leq \vp_1$, $x_*$ and $z_*$ also take imaginary values. We plotted Im$\frac{2x_*}{a}$ and 
$\Delta \vp$ as a function of $s=\frac{z_*}{a}$ for imaginary values of $s$ in Fig.\ref{fig:enchiplot}.
Both get maximized at $s\simeq 1.937i$, where we have 
$\Delta\vp=\vp_1\simeq 2.858$ and $X_T=1.002i$. Thus there are two branches of the solutions. We choose the one which is connected to $s=0$ in the present paper as this is directly connected to our real valued solution for $\Delta\vp\leq \vp_c$.

\begin{figure}
    \centering
    \includegraphics[width=4cm]{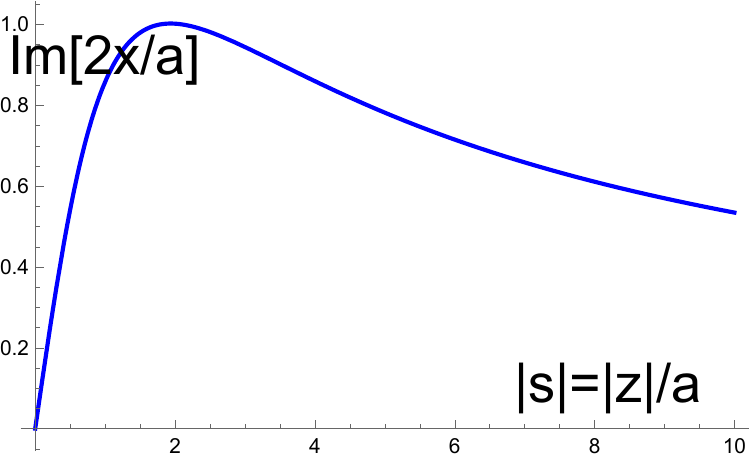}
        \includegraphics[width=4cm]{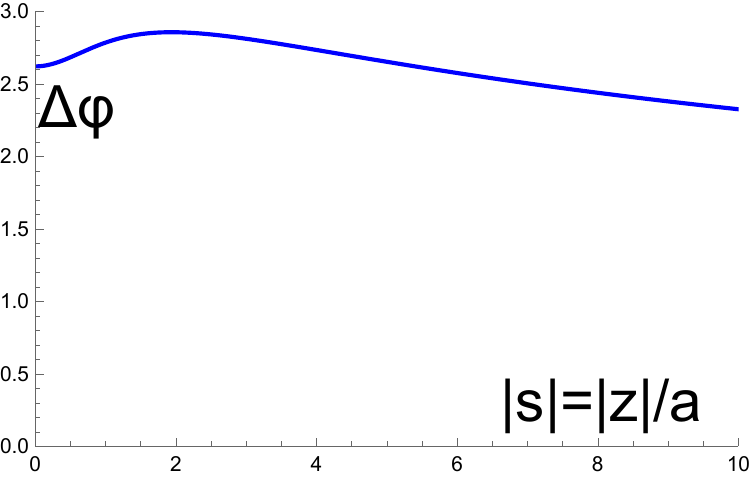}
    \caption{Plots of Im$\frac{2x_*}{a}=$Im$X_T(s)$ (left) and $\Delta \vp=\Phi_T(s)$ as a function of $|s|$ for imaginary $s$.}
    \label{fig:enchiplot}
\end{figure}

\subsection{BTZ for $0\leq \Delta\vp\leq \vp_1$} \label{ap:btzinv}

In the BTZ phase, we find the real valued functions 
$\frac{2x_*}{a}=X_B(s)$ and $\Delta \vp=\Phi_B(s)$ for $0\leq s\leq 1$ given by 
Eq.\eqref{solbtz}.  
They are
plotted in Fig.\ref{fig:enchiplotb}. We find $\Delta\vp=0$ and $x_*=0$ at $s=1$, while we have 
$\Delta\vp=\vp_c$ and $x_*=0$ at $s=0$. Also $\Delta\vp$ takes its maximum $\Delta\vp=\vp_1\simeq 2.858$ at $s\simeq 0.889$.
Note that we can find two solutions for $\vp_c<\vp<\vp_1$: The one is $s \lesssim 0.889$ and the other is $s \gtrsim 0.889$. 
One can find that the former always has larger free energy than the latter, so that we always take the latter solution as a BTZ phase (see Fig.\ref{fulx} and \ref{fula}). 
 
\begin{figure}
    \centering
    \includegraphics[width=4cm]{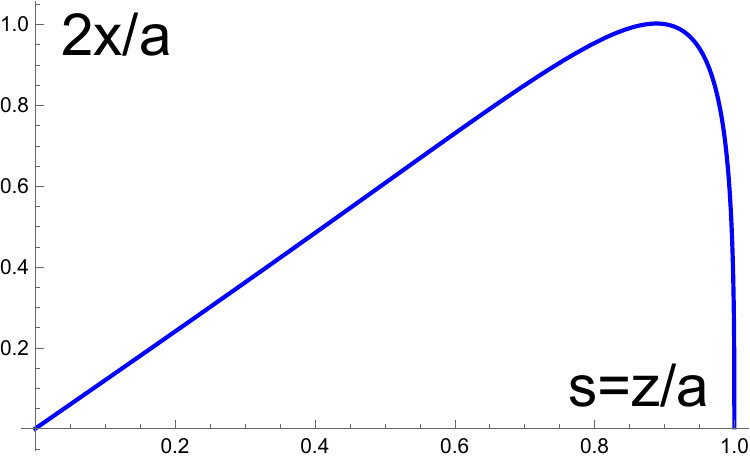}
        \includegraphics[width=4cm]{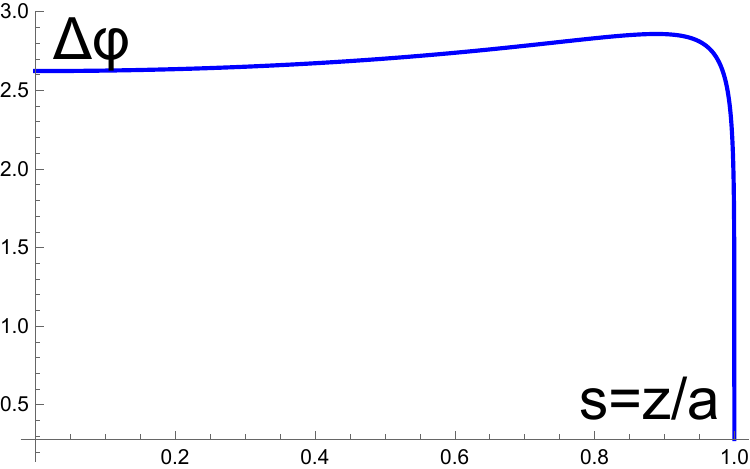}
    \caption{Plots of $\frac{2x_*}{a}=X_B(s)$ (left) and $\Delta \vp=\Phi_B(s)$ (right) as a function of $s$ for $0\leq s\leq 1$.}
    \label{fig:enchiplotb}
\end{figure}

\subsection{Thermal AdS for $\Delta\vp> \vp_1$}

\begin{figure}[t]
    \centering
    \includegraphics[scale = 0.35]{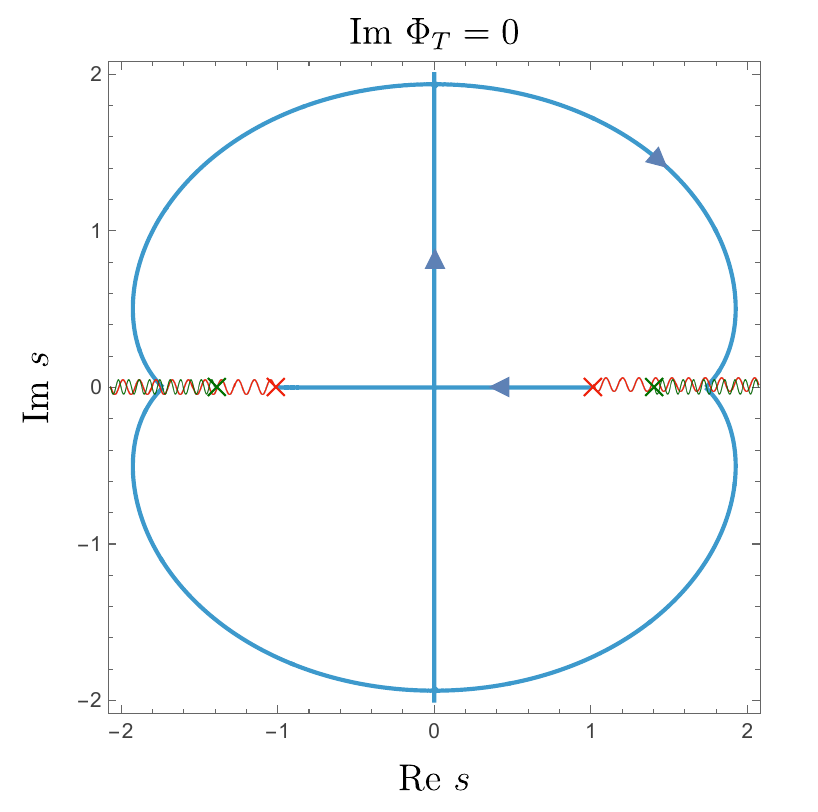}
    \caption{The contour plot of $\Im \Phi_T = 0$. 
      The contour extends along the real axis from $-1$ to $1$, the full imaginary axis, and also the non-trivial curve in each quadrant. 
      To monotonically increase the value of $\Phi_T$, we start from $s = 1$, pass through $s = 0$ and $s = 1.937 i$, and follow the route in first quadrant. 
      As we saw, $\Phi_T (s = 1) = 0$, $\Phi_T (s = 0) = \varphi_c$, and $\Phi_T (s \simeq 1.937 i) = \varphi_1$. 
      We also note that the function has fourth-root branch points at $s = \pm 1$, and logarithmic ones at $s = \pm \sqrt{2}$. 
      Each branch cut runs along the real axis to infinity. 
      }
      \label{phitcomp1}
\end{figure}

\begin{figure}[t]
    \centering
    \includegraphics[scale = 0.35]{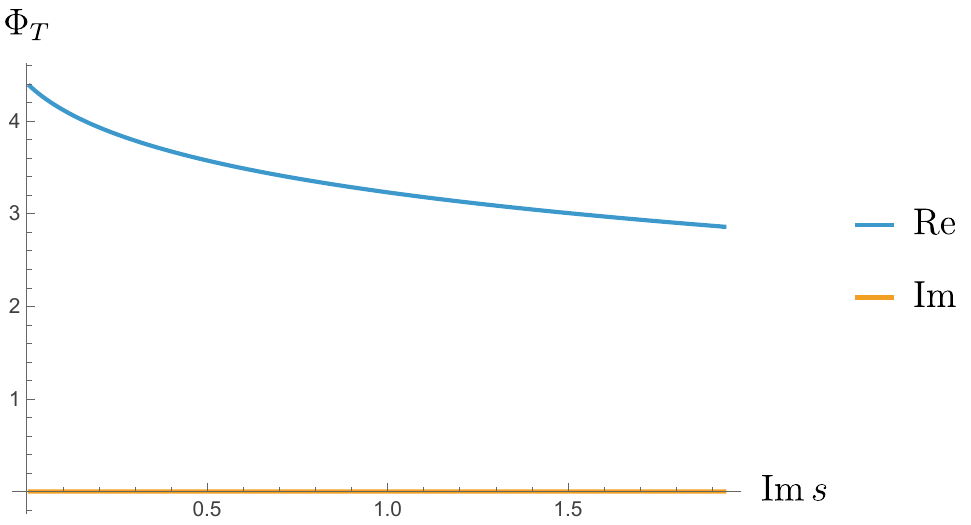}
    \caption{Evaluation of $\Phi_T$ along the curve in first quadrant. 
      The horizontal axis is the imaginary part of $s$, which monotonically decreases along the curve. 
      The blue one shows $\Re \Phi_T$ and orange one is $\Im \Phi_T$. 
      We can explicitly check that indeed $\Phi_T$ take real values which exceed $\varphi_1$.
      }
      \label{phitcomp2}
\end{figure}

We can further extend the phase diagram for $\Delta \varphi$ larger than $\varphi_1$. 
To see this, we first scrutinize $\Phi_{T} (s)$ as a complex function of $s$. 
Since we are assuming $\Delta \varphi$ to be real, we focus on the contour in the $s$-plane on which $\Im \Phi_T = 0$ (Fig.\ref{phitcomp1}). 
It shows that the contour extends along the real axis from $-1$ to $1$, the full imaginary axis, and also the non-trivial curve in each quadrant. 
As we saw in the previous subsections, $\Phi_T$ starts from $0$ at $s = 1$ and monotonically increases to $\Phi_T = \varphi_c \simeq 2.622$ at $s = 0$. 
It continues to increase along $s\in i \mathbb{R}_{\geq 0}$ and reaches to the local maximum $\Phi_T = \varphi_1 \simeq 2.858$ at $s \simeq 1.937 i$. 
To extend $\Phi_T$ larger than $\varphi_1$, we need to follow one of the non-trivial curves. 
Here we pick the one in the first quadrant without loss of generality, and show the value of $\Phi_T$ along the curve in first quadrant in Fig.\ref{phitcomp2}. 
It shows that we can indeed extend the value of boundary scalar source until $\Delta \varphi \simeq 4.410$. 
Since the $\Phi_T$ becomes a monotonic function along the curve, we can make the inverse function of it and evaluate the $X_T[\Phi_T^{-1}]$ and on-shell action (Fig.\ref{fulx} and \ref{fula}).

\subsection{BTZ for $\Delta\vp> \vp_1$}

\begin{figure}[t]
    \centering
    \includegraphics[scale = 0.35]{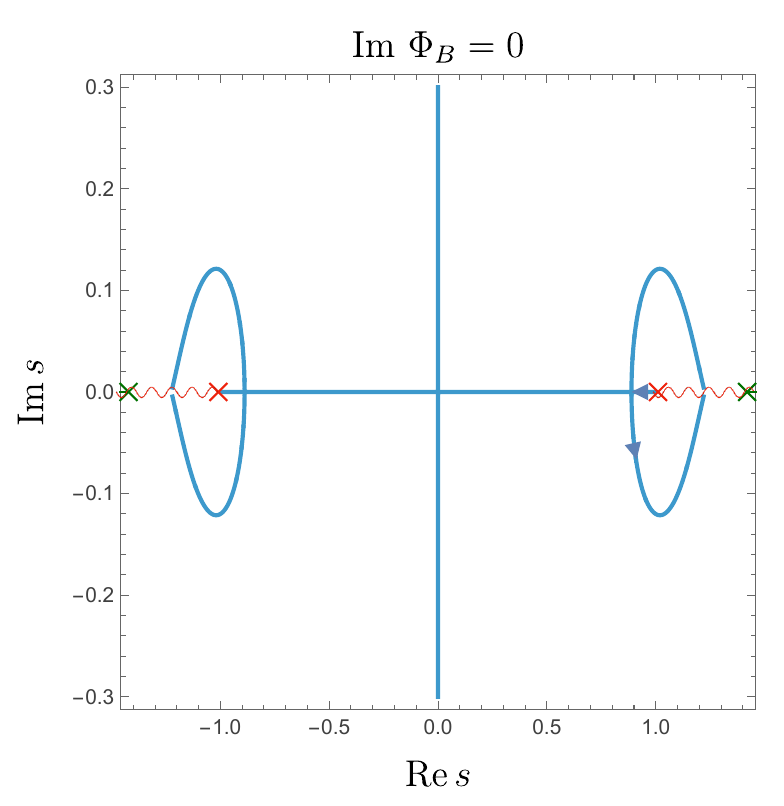}
    \caption{The contour plot of $\Im \Phi_B = 0$. 
      The contour extends along the real axis from $-1$ to $1$, and the non-trivial curve in each quadrant. 
      To monotonically increase the value of $\Phi_B$, we start from $s = 1$, pass through $s = 0$, and follow the route in fourth quadrant. 
      As we saw, $\Phi_B (s = 1) = 0$ and $\Phi_B (s \simeq 0.889) = \varphi_1$.
      We note that the function has fourth-root branch points at $s = \pm 1$, logarithmic ones at $s = \pm 1$ and $\pm \sqrt{2}$, and square root branch at $s = \pm \sqrt{2}$. 
      Each fractional branch cut extends along the real axis to infinity, while the logarithmic cut runs between $[- \sqrt{2}, -1]$ and $[1, \sqrt{2}]$. 
      }
      \label{phibcomp1}
\end{figure}

\begin{figure}[t]
    \centering
    \includegraphics[scale = 0.35]{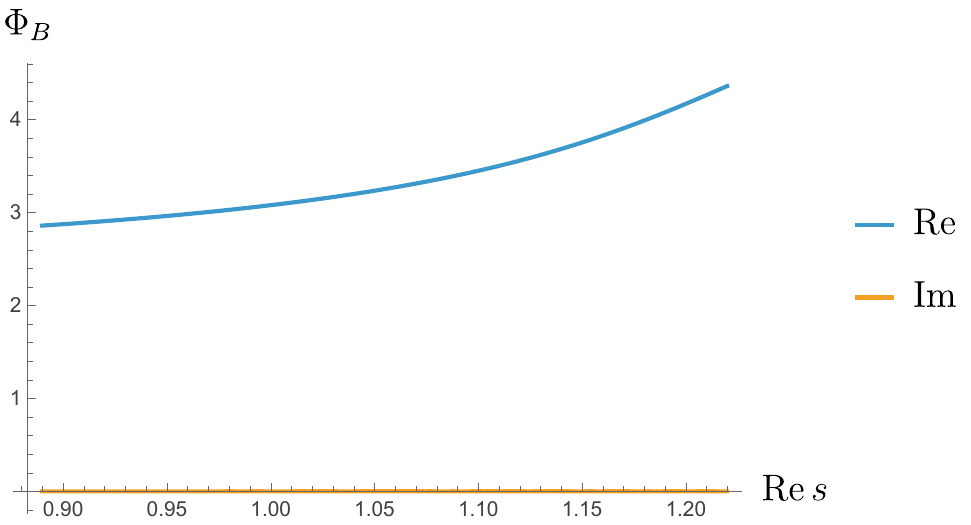}
    \caption{Evaluation of $\Phi_B$ along the curve in fourth quadrant. 
    The horizontal axis is the real part of $s$, which monotonically increases along the curve. 
    The blue one shows $\Re \Phi_B$ and orange one is $\Im \Phi_B$. 
    We can explicitly check that indeed $\Phi_B$ take real values which exceed $\varphi_1$.}
    \label{phibcomp2}
\end{figure}

We take the same procedure to $\Phi_B$. 
The contour of $\Im \ \Phi_B = 0$ is shown in Fig.\ref{phibcomp1}. 
The new branch to exceed $\varphi_1$ emerges at $s \simeq 0.889$. 
Along the non-trivial curve in the fourth quadrant, $\Phi_B$ behaves like in Fig.\ref{phibcomp2}. 
It shows that we can indeed extend the value of boundary scalar source until $\Delta \varphi \simeq 4.410$. 
$\Phi_B$ is also a monotonic function along the curve, so we can make the inverse function of it and evaluate the $X_B[\Phi_B^{-1}]$ and on-shell action. 
They are shown in Fig.\ref{fulx} and \ref{fula} respectively. 
One can check that the real part of $X_B$ is always positive at this range, so that the real part of on-shell action is always smaller than that of disconnected BTZ phase. 
We also remark that the first term of connected BTZ on-shell action \eqref{actions} corresponds to thermal free energy of the black hole metric $\beta F = \beta E - S$, while the second one comes from the boundary effects.
The imaginary value of on-shell action fully comes from the latter, as we expect from the initial setup.

Combining the results of this section, we can draw phase diagram of this system with wider range of $\Delta \varphi$.
The results are shown in Fig.\ref{fulaa}.

\begin{figure}[t]
  \centering
  \includegraphics[scale = 0.3]{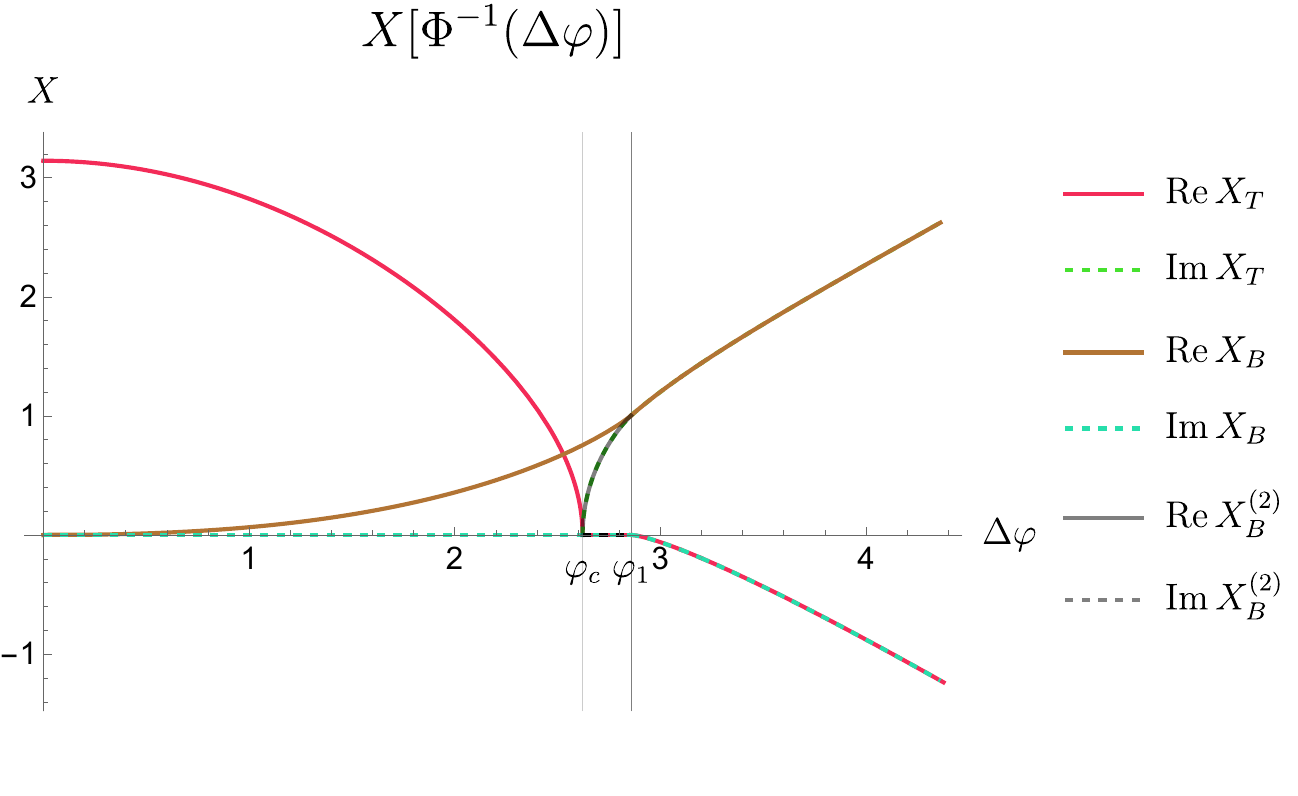}
  \caption{$X [ \Phi^{-1} (\Delta \varphi)]$ for each solution. $X_T$ represents a solution for TAdS, and $X_B$ is for BTZ. $X_B^{(2)}$ is a sub-dominant solution of BTZ with $s \lesssim 0.889$ which we mentioned in the end of Appendix.\ref{ap:btzinv}. } 
  \label{fulx}
\end{figure}

\begin{figure}[t]
  \centering
  \includegraphics[scale = 0.3]{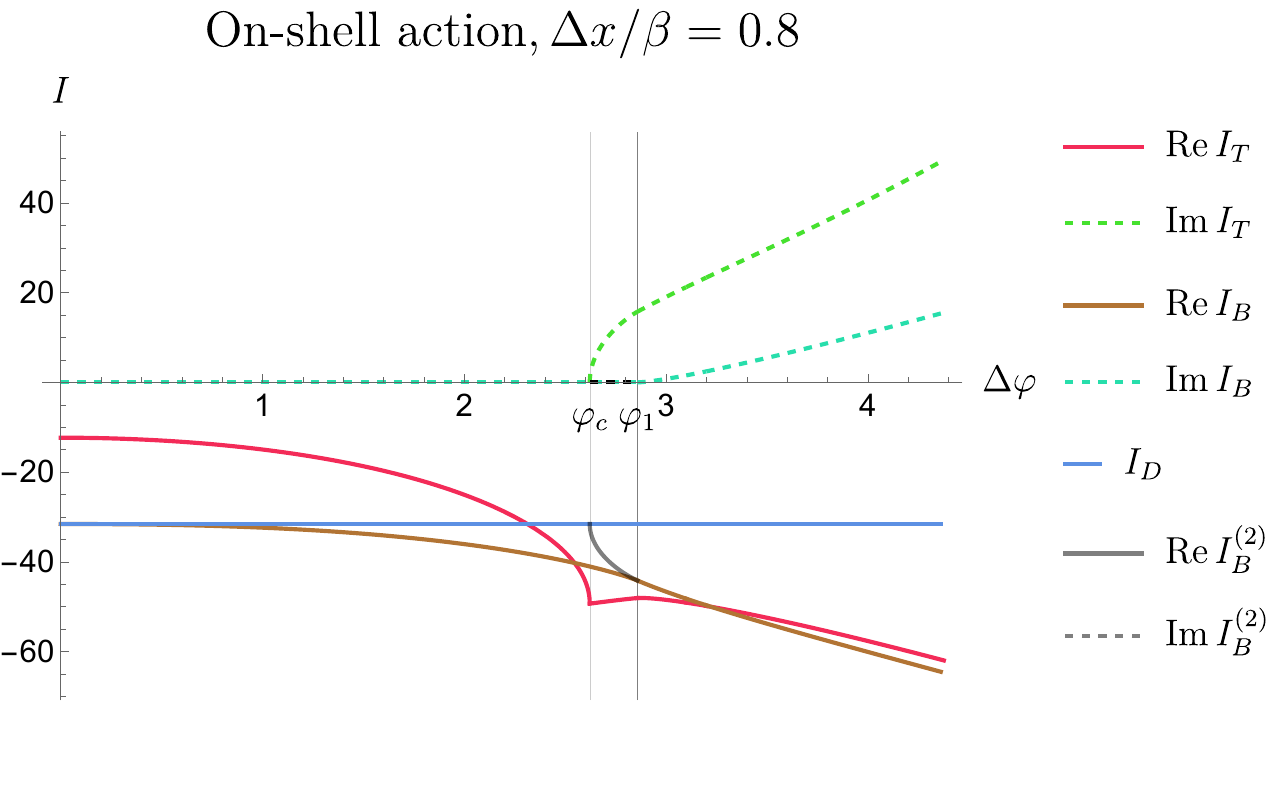}
  \caption{On-shell actions for each solution. $I_T$ is an on-shell action of TAdS solution, $I_B$ is that of BTZ with connected brane, and $I_D$ is BTZ with disconnected brane. $I_B^{(2)}$ is a sub-dominant branch of BTZ with $s \lesssim 0.889$ which we mentioned in the end of Appendix.\ref{ap:btzinv}. We set $16 \pi G_N = 1$, and $\Delta x / \beta = 0.8$. 
  We can easily check that the $I_D$ and $I_B^{(2)}$ are always sub-dominant. }
  \label{fula}
\end{figure}

\begin{figure}[t]
  \centering
  \includegraphics[scale = 0.3]{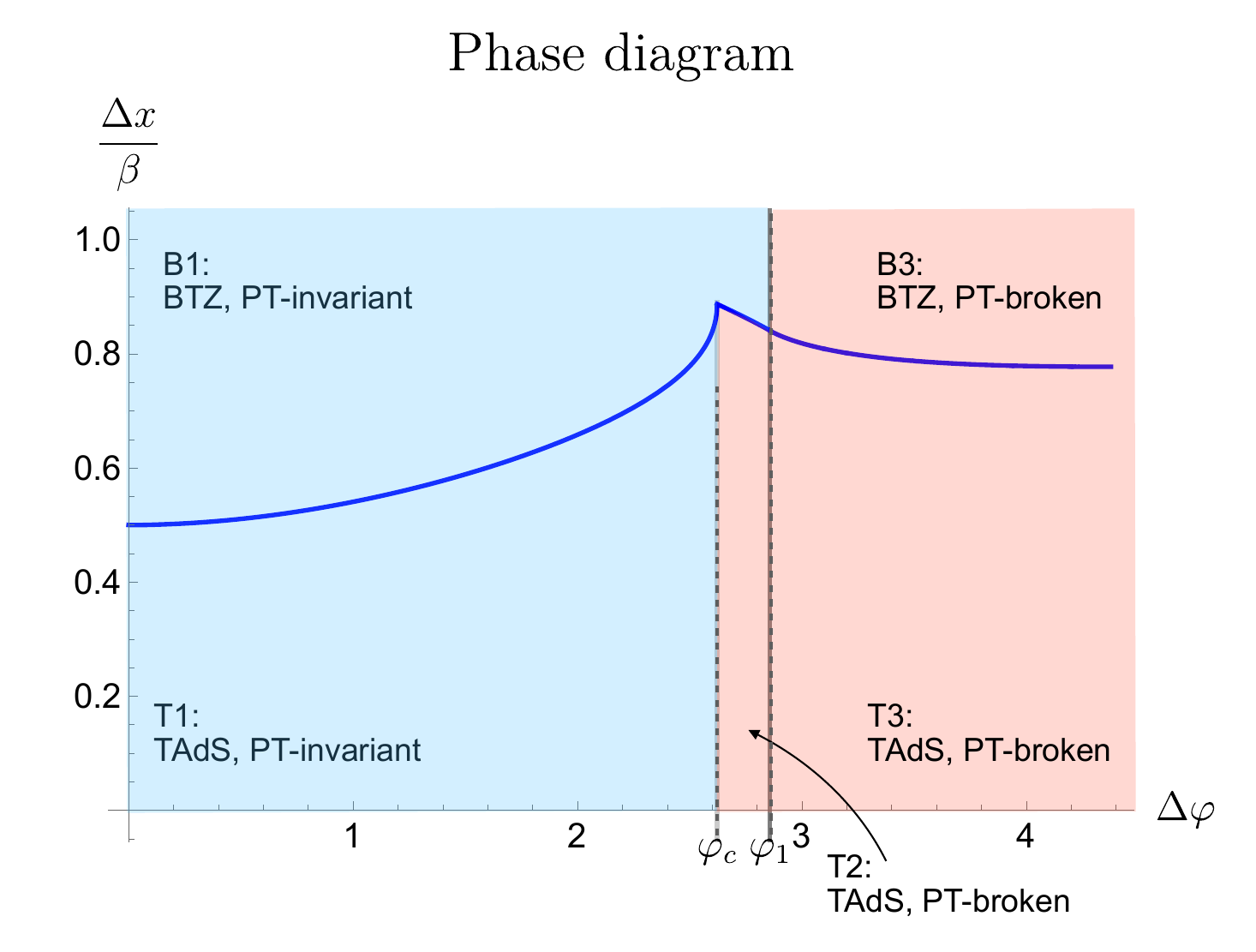}
  \caption{A phase diagram extended to $\Delta \varphi \lesssim 4.410$. }
  \label{fulaa}
\end{figure}

\subsection{Further extension for $\Delta \varphi$}
In the previous subsection we extended the value of $\Delta \varphi$ up to $\sim 4.410$. 
One may worry that a new phase transition might occur at that point, but it is not the case. 
To see that, we will further extend the phase diagram, which demands us to go beyond the branch cuts and consider the second Riemann surface. 
Following monodromy conditions of elliptic integrals are useful: 
\begin{align}
  K(1 + e^{\pm 2 \pi i} (m - 1)) &= K(m) \mp 2 i K(1-m), \label{kmono} \\
  \Pi(1 + e^{\pm 2 \pi i} (n-1), m) &= \Pi (n, m) \mp \frac{i \pi}{\sqrt{n-1} \sqrt{1 - m / n}}, \label{p1mono}
\end{align}
\begin{equation}
    \begin{aligned}
        &\Pi(n, 1 + e^{\pm 2 \pi i} (m-1)) \\ 
        &= \Pi(n, m) \mp 2 i \left[ K(1-m) + \frac{n}{1-n} \Pi \left( \frac{1-m}{1-n}, 1-m \right) \right]. \label{p2mono}
    \end{aligned}
\end{equation}
For example, we can evaluate $\Phi_T(s)$ beyond the branch cut at $s \geq \sqrt{2}$ as
\begin{equation}
  \begin{aligned}
    \Phi_T^{\mathrm{2 nd.}} (s) &= \left[ 2(1-s^2)^{\frac{1}{4}} K \left( -(1-s^2) \right) \right]^{\mathrm{2nd.}} \\
    &= \mp 2 i (1-s^2)^{\frac{1}{4}} \left[ K \left( -(1-s^2) \right) \pm 2 i K(2 - s^2) \right]
  \end{aligned} \label{ftbey}
\end{equation}
Here the superscript $^{\mathrm{2nd.}}$ means that the function is evaluated on the second Riemann surface. 
One should take upper sign when one crosses the branch cut from first quadrant to fourth, and lower sign for the other direction. 
In the same way
\footnote{
    Note that $X_T$ has square root branches at $s = 1$, which come from the pre-factor $\sqrt{1 - s^2}$ and Eq.\eqref{p1mono}, and logarithmic branch at $s = \sqrt{2}$ from Eq.\eqref{p2mono}.
    In the same way, $X_B$ has square root branches at $s = 1, \sqrt{2}$, coming from the pre-factor, a square root branch at $s = \sqrt{2}$ from Eq.\eqref{p1mono}, and logarithmic branch at $s = 1$ from Eq.\eqref{p2mono}. 
    branch cuts run from each branch point to infinity along the positive real axis. 
    The same structure appears on the negative real axis, because $s$ dependence always appears only through $s^2$. 
}
\begin{align}
  &\begin{aligned}
    X_T^{\mathrm{2nd.}} (s) = - \frac{2 \sqrt{1 - s^2}}{s} \Bigg[ \Pi \left( s^2, s^2 - 1 \right) \hspace{50pt} \\ 
    \pm 2 i \left\{ K(2-s^2) + \frac{s^2}{1 - s^2} \Pi \left( \frac{2-s^2}{1-s^2}, 2-s^2 \right) \right \} \\
    \pm \frac{i \pi}{\sqrt{1 - 1/s^2}} - K(s^2-1) \mp 2i K(2-s^2)
    \Bigg] &, \\ 
    \qquad (s \geq \sqrt{2}),
  \end{aligned} \\
  &\begin{aligned}
      \Phi_B^{\mathrm{2nd.}} (s) = \mp \frac{2 i (1-s^2)^{\frac{1}{4}}}{\sqrt{2-s^2}} \left[  K \left( \frac{1}{2 - s^2} \right) \pm 2 i K \left( \frac{1 - s^2}{2 - s^2}  \right) \right], \\
      \qquad (1 \leq s \leq \sqrt{2}), 
  \end{aligned}\\
  &\begin{aligned}
    X_B^{\mathrm{2nd.}} (s) = - \frac{2 s}{ \sqrt{1 - s^2} \sqrt{2 - s^2}} \Bigg[ \Pi \left(\frac{1-s^2}{2-s^2}, \frac{1}{2-s^2} \right) \hspace{10pt} \\ 
    \pm 2 i \left\{ K \left( \frac{1-s^2}{2-s^2} \right) + (1 - s^2) \Pi \left( 1-s^2, \frac{1-s^2}{2-s^2} \right) \right \} & \\
    - K \left( \frac{1}{2-s^2} \right) \mp 2i K \left( \frac{1-s^2}{2-s^2} \right)
    \Bigg] &, \\
    (1 \leq s \leq \sqrt{2}) &. 
  \end{aligned}
\end{align}
Using these formulae, we can easily extend the value of $\Delta \varphi$. 
The extended phase diagram is shown in 
Fig.\ref{fig:phased}. 
We can check the phase transition line extends beyond $\Delta \varphi \gtrsim 4.410 $ smoothly, which impleis the phase transition does not happen there. 
Indeed, we know that the phase transition can occur when two $ \Im \Phi = 0$ contours intersect to each other, which is not the case when we cross the branch cuts. 
This is why we conclude a new phase transition does not occur at that point. 
By repeating similar procedures, we can further and further extend the value of $\Delta \varphi$.

\section{Non-Hermitian quantum mechanics }

Here we briefly review some properties of non-Hermitian systems. 
More comprehensive discussions are given in \textit{e.g.} \cite{Ashida:2020dkc}. 
We focus on the case that the Hamiltonian $H$ is diagonalizable. 

\subsection{Linear Algebra statement}

\begin{dfn}[pseudo-Hermitian]
  A linear operator $H$ is said to be pseudo Hermitian if there exists a Hermitian operator $\tau$ such that
  \begin{equation}
    H^\dagger = \tau H \tau^{-1}. \label{psedef}
  \end{equation}
\end{dfn}

\begin{thm}[pseudo-Hermiticity and spectrum] \label{psespe}
  Assume $H$ is diagonalizable and has a discrete spectrum. 
  Then, the following two are equivalent:
  \begin{enumerate}
    \item $H$ is pseudo-Hermitian.
    \item The eigenvalues of $H$ are either real or come in complex conjugate pairs.
  \end{enumerate}
\end{thm}

\begin{proof}
  ($1. \Rightarrow 2.$) \ We denote the eigenvalue as $E_n$ and right eigenvector as $\ket{n}$. 
  Applying $H^{\dagger}$ to $\tau \ket{n}$, we find
  \begin{equation}
    H^{\dagger} \tau \ket{n} = \tau H \tau^{-1} \tau \ket{n} = E_n \tau \ket{n}. 
  \end{equation}
  Thus, $\tau \ket{n}$ is a right eigenvector of $H^{\dagger}$ with eigenvalue $E_n$.
  Since the set of eigenvalues of $H^{\dagger}$ is the complex conjugate of that of $H$, there exists an eigenvalue $E_m$ of $H$ that satisfies $E_m = E_n^*$.
  Therefore, the eigenvalues of $H$ are either real or come in complex conjugate pairs. 

  ($2. \Rightarrow 1.$) \ We denote the eigenvalues of $H$ as $E_n$, and the corresponding right and left eigenvectors as $\ket{n}$ and $\bra{\tilde{n}}$, respectively. 
  Since $H$ is diagonalizable, we can construct a complete biorthogonal basis
  \begin{equation}
    \braket{\tilde{m}}{n} = \delta_{mn}, \qquad \sum_n \ket{n} \bra{\tilde{n}} = 1. \label{biorth}
  \end{equation}
  Hamiltonian and its Hermitian conjugate are written as
  \begin{equation}
    H = \sum_n E_n \ket{n} \bra{\tilde{n}}, \qquad H^{\dagger} = \sum_n E_n^* \ket{\tilde{n}} \bra{n}.
  \end{equation}
  Here we defined $\bra{n}$ = ($\ket{n}$)$^{\dagger}$, and $\ket{\tilde{n}} = \left( \bra{\tilde{n}} \right)^{\dagger}$. 
  From the assumption, there exists an index $n'$ for each $n$ such that $E_{n'} = E_n^*$. (If $E_n$ is real, we can set $n' = n$.) 
  We define $\tau$ as
  \begin{equation}
    \tau = \sum_n \ket{\tilde{n'}} \bra{\tilde{n}}. 
  \end{equation}
  Choosing the phase of $\ket{n'}$ appropriately, we can make $\tau$ Hermitian. 
  Its inverse is given by
  \begin{equation}
    \tau^{-1} = \sum_n \ket{n} \bra{n'}.
  \end{equation}
  Then, we can check that $H^{\dagger} = \tau H \tau^{-1}$ holds as follows:
  \begin{equation}
    \begin{aligned}
        \tau H \tau^{-1} &= \sum_{l, m, n} E_m \ket{\tilde{l}'} \braket{\tilde{l}}{m} \braket{\tilde{m}}{n} \bra{n'} \\ 
        &= \sum_n E_n \ket{\tilde{n}'} \bra{n'} = H^{\dagger}.
    \end{aligned}
  \end{equation}
\end{proof}

In Eq.\eqref{biorth}, we introduced a biorthogonal system. 
By using this, we can express a trace of a linear operator $A$ as
\begin{equation}
  \Tr(A) = \sum_n \mel{\tilde{n}}{A}{n}.
\end{equation}

\begin{dfn}[quasi-Hermitian]
  A linear operator $H$ is said to be quasi-Hermitian if there exists a positive definite operator $\Theta$ such that 
  \begin{equation}
    H^{\dagger} = \Theta H \Theta^{-1}. \label{quasiH}
  \end{equation}
\end{dfn}

\begin{thm}[quasi-Hermiticity and real spectrum]
  Assume $H$ is diagonalizable and has a discrete spectrum. 
  Then, the following two are equivalent:
  \begin{enumerate}
    \item $H$ is quasi-Hermitian.
    \item The eigenvalues of $H$ are all real.
  \end{enumerate}
\end{thm}

\begin{proof}
  ($1. \Rightarrow 2.$) In general, a positive definite linear operator can be decomposed into the form $\Theta = R^{\dagger} R $ where $R$ is an invertible matrix. 
  From the definition of quasi-Hermiticity, we have
  \begin{equation}
    H^{\dagger} = R^{\dagger} R H R^{-1} (R^{\dagger})^{-1} \ \Leftrightarrow \ (R^{\dagger})^{-1} H^{\dagger} R^{\dagger} = R H R^{-1}.
  \end{equation}
  Therefore, $R H R^{-1}$ is a Hermitian operator, and its eigenvalues are all real. 
  Since $R$ is an invertible matrix, $H$ and $R H R^{-1}$ have the same eigenvalues, so the eigenvalues of $H$ are also real. 
  
  ($2. \Rightarrow 1.$) We denote the eigenvalues of $H$ as $E_n$, and the corresponding right and left eigenvectors as $\ket{n}$ and $\bra{\tilde{n}}$, respectively.
  As in the same way to the proof of Thm.\ref{psespe}, we can construct a complete biorthogonal basis and write $H$ and $H^{\dagger}$ as
  \begin{gather}
    \braket{\tilde{m}}{n} = \delta_{mn}, \qquad \sum_n \ket{n} \bra{\tilde{n}} = 1, \\
    H = \sum_n E_n \ket{n} \bra{\tilde{n}}, \qquad H^{\dagger} = \sum_n E_n \ket{\tilde{n}} \bra{n}.
  \end{gather}
  constructing $\Theta$ and its inverse as
  \begin{equation}
    \Theta = \sum_n \ket{\tilde{n}} \bra{\tilde{n}}, \qquad \Theta^{-1} = \sum_n \ket{n} \bra{n},
  \end{equation}
  we can easily check that $H^{\dagger} = \Theta H \Theta^{-1}$ holds. 
  Moreover, for any nonzero vector $\ket{x}$, we have
  \begin{equation}
    \ev{\Theta}{x} = \sum_n |\braket{\tilde{n}}{x}|^2 > 0, 
  \end{equation}
  so $\Theta$ is positive definite.
\end{proof}

For a quasi-Hermitian Hamiltonian, we can define a new inner product as $\braket{x}{y}_{\Theta} := \mel{x}{\Theta}{y}$. 
Under this new inner product, the norm of any state is positive, and the time evolution by $H$ becomes unitary  \cite{SCHOLTZ199274}.
Note, however, that $\Theta$ which satisfies Eq.\eqref{quasiH} is not unique, and the choice of the new inner product is also not unique.

\subsection{PT-symmetric quantum mechanics}

We now review the PT-symmetric quantum mechanics \cite{Bender:2005tb}.
We call $H$ PT-invariant if there exists an anti-linear operator $PT$ such that 
\begin{equation}
  (PT) H (PT)^{-1} = H. \label{ptsym}
\end{equation}
In most cases, we also assume that parity operator $P$ is linear and unitary, time reversal $T$ is anti-linear and anti-unitary, which commute to each other and involutive;
\begin{equation}
  [P, T] = 0, \quad P^2 = 1, \quad T^2 = 1.
\end{equation}

PT-symmetry implies some nice properties. 
One is that the eigenvalues of $H$ are either real or come in complex conjugate pairs, \textit{i.e.} $H$ is pseudo-Hermitian.

\begin{thm}[PT-invariance and pseudo-Hermiticity]
  Assume $H$ is diagonalizable and has a discrete spectrum. 
  If $H$ is PT-invariant, then $H$ is pseudo-Hermitian.
\end{thm}

\begin{proof}
  We prove that the eigenvalues of PT-invariant Hamiltonian are real or come in complex conjugate pair. 
  Make $E_n$ be an eigenvalue and $\ket{n}$ be a right eigenvector. 
  If $E_n$ is real, there is nothing to show, so we assume $E_n$ is complex.
  Applying $H$ to a vector $PT \ket{n}$, we get
  \begin{equation}
    \begin{aligned}
      H (PT \ket{n}) &= PT H (PT)^{-1} PT \ket{n} \\
      &= PT E_n \ket{n} = E_n^* PT \ket{n} \label{ptn}
    \end{aligned}  
  \end{equation}
  Thus, $E_n^*$ is also an eigenvalue of $H$, and the eigenvalues of $H$ are either real or come in complex conjugate pairs.
\end{proof}

Note that the opposite is also true: if $H$ is pseudo-Hermitian, there exists an anti-linear operator $PT$ which satisfies Eq.\eqref{ptsym} \cite{Mostafazadeh:2002id}.

From Eq.\eqref{ptn}, we see that after applying the PT operation to the eigenvector $\ket{n}$ with real eigenvalue $E_n$, the energy eigenvalue remains the same one. 
In particular, if there is no degeneracy, the eigenstate is invariant. 
\begin{equation}
    PT \ket{n} \propto \ket{n}
\end{equation}
In general, we call the system to be PT-invariant when all of its eigenvalues are real, \textit{i.e.}, quasi-Hermitian. 
On the other hand, it is said to be PT-broken when any of the eigenvalues takes complex value. 

Combining the $PT$ operator and the operator $\tau$, we can make a new anti-linear operator $V$ which relates the right eigenvector to the Hermitian conjugate of the left eigenvector for the same eigenvalue. 

\begin{thm}[PT-symmetry and $V$-operator]
  Assume $H$ is diagonalizable and has a discrete spectrum. 
  If $H$ is PT-invariant, then there exists an anti-linear operator $V$ such that
  \begin{equation}
    V H V^{-1} = H^{\dagger}, \qquad V \ket{n} \propto \ket{\tilde{n}}. \label{vhvh}
  \end{equation}
\end{thm}

\begin{proof}
  Choosing $V$ as $V = \tau P T$, the statements follow immediately. 
  \begin{gather}
    V H V^{-1} = \tau P T H (P T)^{-1} \tau^{-1} = \tau H \tau^{-1} = H^{\dagger}, \\
     V \ket{n} = \tau P T \ket{n} \propto \tau \ket{n'} = \ket{\tilde{n}}.
  \end{gather}
\end{proof}

\subsection{Example: two sites model}
Consider the Hamiltonian in the two level system
\ba
H=
\left(
\begin{array}{cc}
i\gamma & -W \\
-W & -i\gamma \\
\end{array}
\right).
\ea
Here we assume $\gamma > 0$ and $W > 0$. 
This is invariant under the PT transformation with 
\ba
P=
\left(
\begin{array}{cc}
0& 1 \\
 1 & 0 \\
\end{array}
\right),
\ea
and $T$ acts as the complex conjugation.

Two right eigenvectors are
\footnote{
    We promise that for $\gamma > W$, we choose $\sqrt{W^2 - \gamma^2} = + i \sqrt{\gamma^2 - W^2}$ sign. 
}
\begin{align}
  \ket{n_1} = \frac{1}{\sqrt{2} W} \begin{pmatrix}
     W \\ i \gamma + \sqrt{W^2 - \gamma^2}
  \end{pmatrix} , \qquad E_1 = - \sqrt{W^2 - \gamma^2} \\
  \ket{n_2} = \frac{1}{\sqrt{2} W} \begin{pmatrix}
     W \\ i \gamma - \sqrt{W^2 - \gamma^2}
  \end{pmatrix}, \qquad E_2 = + \sqrt{W^2 - \gamma^2}
\end{align}
Thus, the system is PT-invariant when $W > \gamma$, while it is PT-broken for $W < \gamma$.
Note that when $W = \gamma$, not only two eigenvalues but also two eigenvectors are degenerate, so that Hamiltonian is not diagonalizable. 
It is called the exceptional point, and two phases (PT-invariant and PT-broken) are separated by this point. 
Note that corresponding left eigenvectors are respectively
\begin{align}
  \bra{\tilde{n}_1} &= \frac{1}{\sqrt{2 (W^2 - \gamma^2)}} \begin{pmatrix}
    - i \gamma + \sqrt{W^2 - \gamma^2}, & W
  \end{pmatrix}, \\
  \bra{\tilde{n}_2} &= \frac{1}{\sqrt{2  (W^2 - \gamma^2)}} \begin{pmatrix}
    i \gamma + \sqrt{W^2 - \gamma^2}, & - W
  \end{pmatrix},  
\end{align}
which indeed satisfy the biorthonormal condition \eqref{biorth}.

\subsubsection{PT-invariant phase}
For $W > \gamma$ case, $n_1' = n_1$ and $n_2' = n_2$. 
We can construct $\Theta$ operator as
\begin{equation}
  \begin{aligned}
    \Theta &= \ket{\tilde{n}_1} \bra{\tilde{n}_1} + \ket{\tilde{n}_2} \bra{\tilde{n}_2} \\
    &= \frac{W}{W^2 - \gamma^2} \begin{pmatrix}
       W & i \gamma \\ - i \gamma & W
    \end{pmatrix}. 
  \end{aligned}
\end{equation}
We can easily check that $\Theta$ indeed satisfies Eq.\eqref{quasiH}, and its eigenvalues are $\frac{W (W \pm \gamma)}{W^2 - \gamma^2} > 0$, so $\Theta$ is positive definite. 
$V$ operator is defined as 
\begin{equation}
  V = \Theta P T = \frac{W}{W^2 - \gamma^2} \begin{pmatrix}
     i \gamma & W \\ W & - i \gamma
  \end{pmatrix} T . 
\end{equation}
We can check Eq.\eqref{vhvh} as follows:
\begin{align}
  &\begin{aligned}
    V \ket{n_1} 
    &= \frac{1}{\sqrt{2 (W^2 - \gamma^2)}} \begin{pmatrix}
      W \\ -i \gamma + \sqrt{W^2 - \gamma^2}
    \end{pmatrix} \\
    &= \frac{-i \gamma + \sqrt{W^2 - \gamma^2}}{W} \ket{\tilde{n}_1}
  \end{aligned}\\
  &\begin{aligned}
    V \ket{n_2} &= \frac{1}{\sqrt{2 (W^2 - \gamma^2)}} \begin{pmatrix}
      - W \\ i \gamma + \sqrt{W^2 - \gamma^2}
    \end{pmatrix} \\
    &= \frac{-i \gamma - \sqrt{W^2 - \gamma^2}}{W} \ket{\tilde{n}_2}
  \end{aligned}
\end{align}

\subsubsection{PT-broken phase}
For $W < \gamma$ case, $n_1' = n_2$ and $n_2' = n_1$ respectively. 
We can construct $\tau$ operator as
\begin{equation}
  \begin{aligned}
    \tau &= \ket{\tilde{n}_2} \bra{\tilde{n}_1} + \ket{\tilde{n}_1} \bra{\tilde{n}_2} \\
    &= \frac{W}{\gamma^2 - W^2} \begin{pmatrix}
      -W & - i \gamma \\ i \gamma & - W
    \end{pmatrix}. 
  \end{aligned}
\end{equation}
This operator indeed satisfies Eq.\eqref{psedef}, and its eigenvalues are $\frac{W (-W \pm \gamma)}{\gamma^2 - W^2}$. 
One is positive and the other is negative, so that $\tau$ is not positive definite. 
$V$ operator is defined as 
\begin{equation}
  V_{\mathrm{b}} = \tau P T = \frac{W}{\gamma^2 - W^2} \begin{pmatrix}
    - i \gamma & - W \\ - W & i \gamma
  \end{pmatrix} T . 
\end{equation}
We can explicitly check Eq.\eqref{vhvh} as
\begin{align}
  &\begin{aligned}
    V_{\mathrm{b}} \ket{n_1} 
    &= \frac{1}{\sqrt{2 (\gamma^2 - W^2)}} \begin{pmatrix}
      i W \\ \gamma + \sqrt{\gamma^2 - W^2}
    \end{pmatrix} \\
    &= - i \frac{ \gamma + \sqrt{\gamma^2 - W^2}}{W} \ket{\tilde{n}_1}
  \end{aligned}\\
  &\begin{aligned}
    V_{\mathrm{b}} \ket{n_2} &= \frac{1}{\sqrt{2 (\gamma^2 - W^2)}} \begin{pmatrix}
      - i W \\ - \gamma + \sqrt{\gamma^2 - W^2}
    \end{pmatrix} \\
    &= i \frac{- \gamma + \sqrt{\gamma^2- W^2}}{W} \ket{\tilde{n}_2}
  \end{aligned}
\end{align}

\subsection{Path-integral and partition function}
Consider the Euclidean path-integral in the setup of non-Hermitian system. 
For a scalar field $\phi$, the partition function based on  Euclidean path-integral looks like
\ba
Z(\beta)=\int d\phi\la \phi|e^{-\beta H}|\phi\lb. 
\ea
We can rewrite this by using $\sum_n |n\lb\la\tilde{n}|=1$ as
\ba
Z(\beta)&=&\sum_{n,m}\int d\phi \la \phi|n\lb\la \tilde{n}|e^{-\beta H}|m\lb\la\tilde{m}|\phi\lb\no
&=&\sum_n \la \tilde{n}|e^{-\beta H}|n\lb\no
&=&\sum_n e^{-\beta E_n}. \label{zebe}
\ea
Thus it coincides with $\mbox{Tr}[e^{-\beta H}]$. 
This argument is applicable to any non-Hermitian $H$.
We remark that unlike from usual Hermitian case, this partition function does not necessarily describe thermal equilibrium of the system. 

For PT-symmetric Hamiltonian, we can explicitly check that Eq.\eqref{zebe} is always real, while it might be negative. 
\begin{equation}
  Z (\beta) = \sum_n e^{- \beta E_n} + \sum_m 2 e^{- \beta \, \Re E_m} \cos(\beta \, \Im E_m) \label{zpts}
\end{equation}
Here we labeled the real eigenenergies by $n$ and complex conjugate pairs by $m$. 
One may worry that in the main text, we got a free energy with complex values, which implies $Z = e^{- \beta F}$ also becomes complex. 
However, this is not the case. 
A complex free energy arises from the solution of Einstein equation with complex metric. 
We chose one solution, but the complex conjugate of it also satisfies the equation. 
We expect these two solutions contribute to Euclidean path integral with the same weight, as in Eq.\eqref{zpts}. 
Therefore the reality of partition function also remains in our holographic model.


\bibliography{PT.bib}


\end{document}